\documentclass[12pt, draftclsnofoot, onecolumn]{IEEEtran}
\IEEEoverridecommandlockouts
\usepackage{cite}
\usepackage{enumitem}
\usepackage{amsmath,amssymb,amsfonts,amsthm}
\usepackage{eqparbox}
\usepackage{graphicx}
\usepackage{subfigure}
\usepackage{caption}
\usepackage{textcomp}
\usepackage{scalerel}
\usepackage{gensymb}
\usepackage{xcolor}
\usepackage{stfloats}
\usepackage{fancyhdr}
\usepackage{algorithm}
\usepackage{algorithmicx}
\usepackage{algpseudocode}
\newcommand{\algrule}[1][.2pt]{\par\vskip.3\baselineskip\hrule height #1\par\vskip.3\baselineskip}

\newtheorem{lemma}{Lemma}
\newtheorem{defn}{Definition}
\newtheorem{thm}{Theorem}
\newtheorem{Prop}{Proposition}
\newtheorem{remark}{Remark}

\usepackage{multirow}
\newcolumntype{P}[1]{>{\centering\arraybackslash}p{#1}}
\newcolumntype{M}[1]{>{\centering\arraybackslash}m{#1}}

\renewcommand{\Re}{\operatorname{Re}}
\renewcommand{\Im}{\operatorname{Im}}

\def\BibTeX{{\rm B\kern-.05em{\sc i\kern-.025em b}\kern-.08em T\kern-.1667em\lower.7ex\hbox{E}\kern-.125em}}

\begin{document}

{\Large \textbf{Notice:} This work has been submitted to the IEEE for possible publication. Copyright may be transferred without notice, after which this version may no longer be accessible.}
\clearpage

\title{Energy-Efficient Power Control and Beamforming for Reconfigurable Intelligent Surface-Aided Uplink IoT Networks\\}
\author{Jiao Wu, Seungnyun Kim, and Byonghyo Shim\\
Institute of New Media and Communications, Department of Electrical and Computer Engineering, Seoul National University, Korea\\
Email: \{jiaowu, snkim, bshim\}@islab.snu.ac.kr
\thanks{This work was supported by the National Research Foundation of Korea (NRF) grant funded by the Korea government (MSIT) (2020R1A2C2102198) and Samsung Research Funding \& Incubation Center for Future Technology of Samsung Electronics under Project Number (SRFC-IT1901-17).

Parts of this paper was presented at the WCNC21 \cite{wu2021power}.}}

\maketitle

\begin{abstract}
Recently, reconfigurable intelligent surface (RIS), a planar metasurface consisting of a large number of low-cost reflecting elements, has received much attention due to its ability to improve both the spectrum and energy efficiencies by reconfiguring the wireless propagation environment.
In this paper, we propose a RIS phase shift and BS beamforming optimization technique that minimizes the uplink transmit power of a RIS-aided IoT network.
Key idea of the proposed scheme, referred to as \emph{Riemannian conjugate gradient-based joint optimization (RCG-JO)}, is to jointly optimize the RIS phase shifts and the BS beamforming vectors using the Riemannian conjugate gradient technique.
By exploiting the product Riemannian manifold structure of the sets of unit-modulus phase shifts and unit-norm beamforming vectors, we convert the nonconvex uplink power minimization problem into the unconstrained problem and then find out the optimal solution on the product Riemannian manifold.
From the performance analysis and numerical evaluations, we demonstrate that the proposed RCG-JO technique achieves $94\%$ reduction of the uplink transmit power over the conventional scheme without RIS.
\end{abstract}
\newpage
\IEEEpeerreviewmaketitle

\section{Introduction}
As a paradigm to embrace the connection of massive number of devices, such as sensors, wearables, health monitors, and smart appliances, internet of things (IoT) has received much attention recently \cite{stoyanova2020survey}.
Since most of IoT devices are battery-limited, saving the device power is crucial for the dissemination of IoT networks. 
Recently, reconfigurable intelligent surface (RIS) has been emerging as a potential solution to enhance the sustainability of the IoT network\cite{WuTowards}.
In a nutshell, RIS is a planar metasurface consisting of a large number of low-cost passive reflecting elements, each of which can independently adjust the phase shift on the incident signal\cite{cao2021reconfigurable,cao2021ai}.
Due to the capability to modify the wireless channel by controlling the phase shift of each reflecting element, RIS offers various benefits;
When the direct link between the IoT device and the base station (BS) is blocked by obstacles, RIS can provide a virtual line-of-sight (LoS) link between them, ensuring the reliable link connection without requiring heavy transmit modules \cite{gong2020toward}.
Due to its simple structure, lightweight, and conformal geometry, RIS can be easily deployed in the desired location.
Also, with a small change in wireless standard, RIS can be integrated into existing communication systems.

Over the years, various efforts have been made to improve the 
energy efficiency of IoT networks\cite{WuIntelligent2019,jia2020reconfigurable,wu2021energy,chu2021novel}.
In \cite{WuIntelligent2019}, a joint beamforming technique to minimize the downlink power consumption of wireless network has been proposed.
In \cite{jia2020reconfigurable}, a joint power control and beamforming scheme for the RIS-aided device-to-device (D2D) network has been proposed.
In \cite{wu2021energy}, the energy-efficient downlink power control and phase shift designs for the Terahertz and MISO systems have been presented.
Also, achievable sum throughput of a RIS-aided wireless powered sensor network (WPSN) has been investigated in \cite{chu2021novel}.
While most of existing works focused on the downlink power control of RIS-aided systems, not much work has been made for the power minimization on the uplink side.
Since IoT devices are battery-limited, it is of importance to develop an energy-efficient uplink power control mechanism based on the RIS technique.

A major problem of the RIS-aided uplink power control is that the constraints on the RIS phase shifts and the BS beamforming vectors are nonconvex.
This is because RIS can only change the phase shift of incident signal so that the passive beamforming coefficients of RIS reflecting elements are subject to the unit-modulus constraints\cite{WuTowards}.
This, together with the fact that the active beamforming vectors of BS are subject to the unit-norm constraints\cite{2003beamforming}, makes it very difficult to find out the proper RIS phase shifts and BS beamforming vectors. 

An aim of this paper is to propose an approach that minimizes the uplink transmit power of a RIS-aided IoT network.
Key idea of the proposed scheme, referred to as \emph{Riemannian conjugate gradient-based joint optimization (RCG-JO)}, is to jointly optimize the RIS phase shifts and the BS beamforming vectors using the Riemannian conjugate gradient (RCG) algorithm.
Specifically, by exploiting the smooth product Riemannian manifold structure of the sets of unit-modulus phase shifts and unit-norm beamforming vectors, we convert the uplink power minimization problem to the unconstrained problem on the Riemannian manifold.
Since the optimization over the Riemannian manifold is conceptually analogous to that in the Euclidean space, optimization tools of the Euclidean space, such as the gradient descent method, can be readily used to solve the optimization problem on the Riemannian manifold\cite{LoungIoT}.
In our approach, we employ the RCG method to find out the RIS phase shifts and BS beamforming vectors minimizing the uplink transmit power of a RIS-aided IoT network.

In recent years, there have been some efforts exploiting the manifold optimization technique for the RIS phase shifts control\cite{wu2021power,yu2019miso,pan2020multicell}.
In \cite{wu2021power}, a Riemannian manifold-based alternating optimization technique for the design of RIS phase shifts and BS beamforming vectors has been proposed.
In this scheme, the RIS phase shifts and the BS beamforming vectors are optimized on the Riemannian manifolds in an alternative fashion.
In \cite{yu2019miso,pan2020multicell}, the manifold optimization-based RIS phase shift control schemes have been proposed to maximize the downlink sum rate.

Our work is distinct from the previous studies in the following aspects:
\begin{itemize}
    \item We propose a Riemannian conjugate gradient-based joint optimization (RCG-JO) algorithm to find out the RIS phase shifts and the BS beamforming vectors minimizing the uplink transmit power of a RIS-aided IoT network.
    RCG-JO jointly optimizes the RIS phase shifts and the BS beamforming vectors on the product Riemannian manifold of the sets of unit-modulus RIS phase shifts and unit-norm BS beamforming vectors.
    \item We propose a channel estimation technique for the uplink RIS-aided IoT networks.
    Specifically, we estimate the partial RIS-aided channel information using the pilot signals received at the active reflecting elements, from which we recover the full RIS-aided channel information using the low-rank matrix completion (LRMC) algorithm.
    \item We present the convergence analysis of RCG-JO. 
    Also, from the numerical experiments, we demonstrate that RCG-JO converges to a fixed point within a few number of iterations.  
    \item We provide the empirical simulation results from which we demonstrate that RCG-JO outperforms the conventional power control schemes by a large margin in terms of the uplink transmit power and computational complexity.
    For example, when compared to the conventional power control scheme without RIS, RCG-JO saves around $94\%$ of the uplink transmit power.
    Even when compared to the semidefinite relaxation (SDR)-based power control scheme, RCG-JO saves more than $44\%$ of the uplink transmit power and achieves around $98\%$ reduction in the computational complexity.
\end{itemize}

The rest of the paper is organized as follows. 
In Section II, we present the system model, the channel estimation technique, and the data transmission procedures.
We also formulate the power minimization problem of the RIS-aided uplink IoT network.
In Section III, we present the proposed RCG-JO technique. 
In Section IV, we analyze the convergence and computational complexity of RCG-JO.
In Section V, we provide the simulation results and then conclude the paper in Section VI.

\emph{Notations}: Lower and upper case symbols are used to denote vectors and matrices, respectively.
The superscripts $(\cdot)^{*}$, $(\cdot)^{\text{T}}$, and $(\cdot)^{\text{H}}$ denote the conjugate, transpose, and conjugate transpose, respectively.
The operation $\odot$ denotes the Hadamard product.
$\|\mathbf{x}\|$ and $\|\mathbf{X}\|_{\textup{F}}$ denote the Euclidean norm of a vector $\mathbf{x}$ and the Frobenius norm of a matrix $\mathbf{X}$, respectively.
$|\mathbf{x}|$ represents a vector of element-wise absolute values of $\mathbf{x}$.
$\text{tr}(\mathbf{X})$ is the trace of $\mathbf{X}$.
$\text{diag}(\mathbf{x})$ and $\text{ddiag}(\mathbf{X})$ form diagonal matrices whose diagonal elements are $\mathbf{x}$ and diagonal elements of $\mathbf{X}$, respectively.
In addition, $\text{blkdiag}(\mathbf{X}_{1},\mathbf{X}_{2})$ denotes a block-diagonal matrix with $\mathbf{X}_{1}$ and $\mathbf{X}_{2}$ on the block-diagonal.

\section{RIS-Aided Uplink IoT System Model} 
In this section, we present the system model, the channel estimation, and the data transmission protocols of RIS-aided uplink IoT networks.
We then formulate the power minimization problem.
\subsection{RIS-Aided Uplink IoT Network}
We consider a single-input multi-output (SIMO) uplink IoT network where $K$ devices with a single antenna transmit signals to the BS equipped with $M$ antennas (see Fig. 1).
In this network, a RIS consisting of a planar array of $N = N_{x}\times N_{y}$ reflecting elements is deployed to assist the uplink transmission. 
For example, each RIS reflecting element can adjust the phase of the incident signal independently using positive-intrinsic-negative (PIN)
diodes\cite{WuTowards}.
The phase shifts of RIS reflecting elements are configured through the dedicated control link.
\begin{figure}[t]
\centering
\includegraphics[width=0.55\columnwidth]{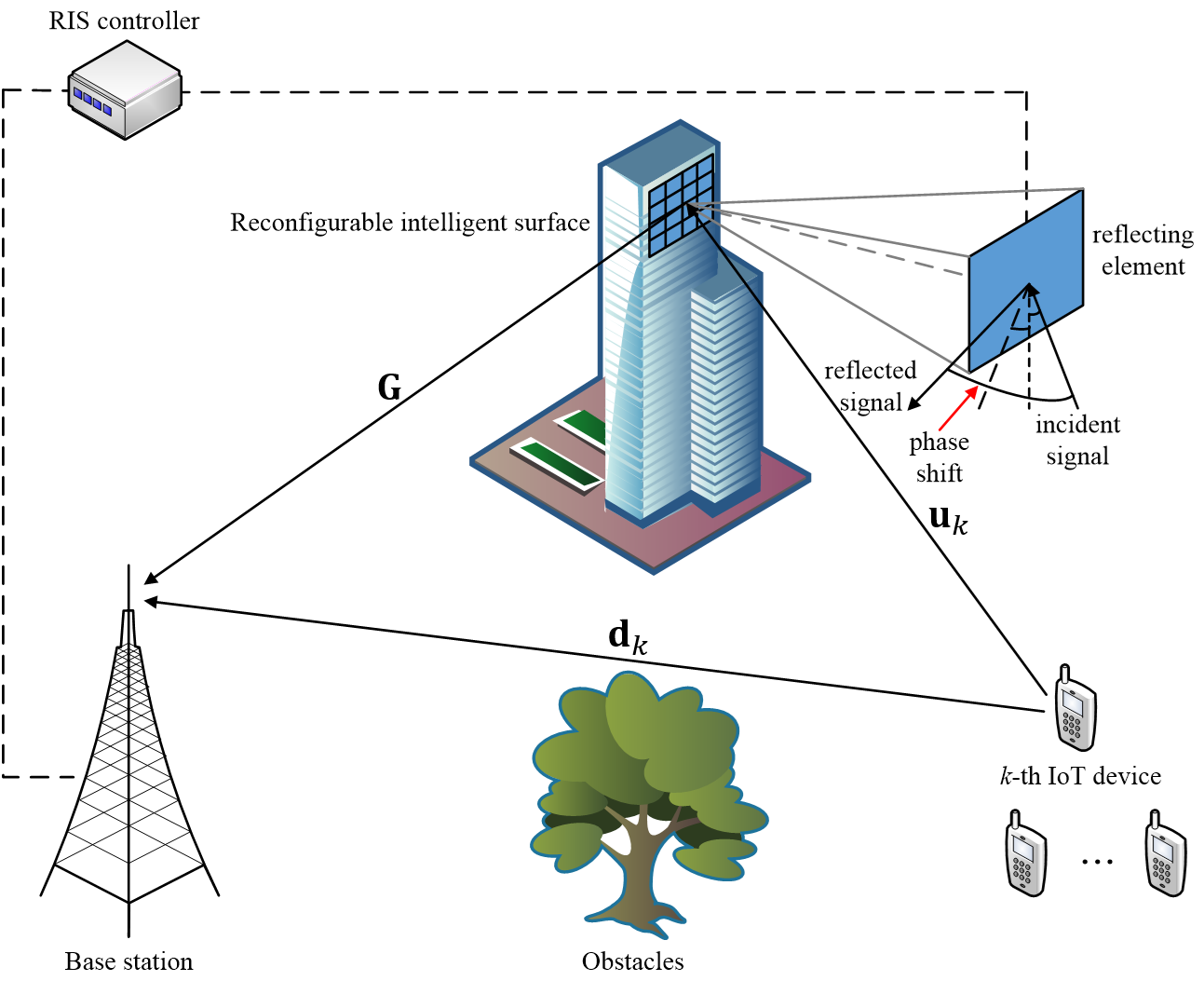}
\caption{Illustration of the RIS-aided uplink IoT network.} \label{fig1}
\end{figure}
In the setup, the effective uplink channel between the $k$-th IoT device and BS is given by 
\begin{equation}
\mathbf{h}_{k} 
= \mathbf{d}_{k} + \mathbf{G}\text{diag}({\boldsymbol{\theta}})\mathbf{u}_{k}
= \mathbf{d}_{k} + \mathbf{G}\text{diag}(\mathbf{u}_{k})\boldsymbol{\theta}
= \mathbf{d}_{k} + \mathbf{G}\mathbf{H}_{k}\boldsymbol{\theta}, \label{1}
\end{equation}
where $\mathbf{d}_{k} \in \mathbb{C}^{M \times 1}$ is the direct channel from the $k$-th IoT device to BS, $\mathbf{u}_{k} \in \mathbb{C}^{N \times 1}$ is the channel from the $k$-th IoT device to the RIS, and $\mathbf{G} \in \mathbb{C}^{M \times N}$ is the channel from the RIS to BS (see Fig. 1).
Also, $\mathbf{H}_{k} = \text{diag}(\mathbf{u}_{k})$ and $\boldsymbol{\theta}=[\mu_{1}\theta_{1},\cdots,\mu_{N}\theta_{N}]^{\text{T}}$ is the phase shift vector of RIS.
In addition, $\mu_{n} \in [0,1]$ is the reflection amplitude coefficient\footnote{In this paper, we assume the ideal phase shift model where the reflection amplitude and phase shift are independent.
In our system, the RIS controller coordinates the switching between two working modes, i.e., receiving mode ($\mu_{n}=0$) for environment sensing (e.g., channel estimation) and reflecting mode ($\mu_{n} = 1$) for scattering the incident signals (e.g., data transmission).} and $\theta_{n}=e^{j\phi_{n}}$ is the passive beamforming coefficient where $\phi_{n} \in [0,2\pi)$ is the phase shift\footnote{To characterize the performance limits of RIS, we assume phase shifts vary continuously in $[0,2\pi)$.
Note that the proposed scheme can be readily extended to practical systems with finite level of phase shifts via discrete phase quantization\cite{QQWtutorial}.}.

In this paper, we assume the Rician fading channel models for $\mathbf{G}$, $\mathbf{u}_{k}$, and $\mathbf{d}_{k}$\cite{he2021reconfigurable,zhang2020capacity}.
First, the channel matrix from the RIS to BS $\mathbf{G}$ is given by
\begin{equation}
    \mathbf{G} = \sqrt{\beta_{\text{G}}} \Big( \sqrt{\frac{\kappa_{\text{G}}}{\kappa_{\text{G}} + 1}}\mathbf{G}^{\textup{LoS}} + \sqrt{\frac{1}{\kappa_{\text{G}} + 1}}\mathbf{G}^{\textup{NLoS}} \Big), \label{2}
\end{equation}
where $\mathbf{G}^{\textup{LoS}}=\mathbf{a}_{\text{BS}}(\vartheta_{\mathbf{G}})\, \mathbf{a}_{\text{RIS}}^{\text{H}}(\psi_{\mathbf{G}},\varphi_{\mathbf{G}})$ is the LoS component with $\vartheta_{\mathbf{G}}$, $\psi_{\mathbf{G}}$, and $\varphi_{\mathbf{G}}$ being the angle of arrival (AoA) at the BS, the azimuth and elevation of the angles of departure (AoDs) at the RIS, respectively, $\mathbf{G}^{\textup{NLoS}}$ is the NLoS component generated from complex Gaussian distribution, $\beta_{\text{G}}$ is the path loss between RIS and BS, and $\kappa_{\text{G}} (\geq 0)$ is the Rician factor.
Here, $\mathbf{a}_{\text{BS}}(\vartheta_{\mathbf{G}}) = [1, \cdots,e^{-j\pi(M-1)\sin\vartheta_{\mathbf{G}}}]^{\text{T}} \in \mathbb{C}^{M \times 1}$ is the BS array steering vector and $\mathbf{a}_{\text{RIS}}(\psi_{\mathbf{G}},\varphi_{\mathbf{G}}) = \mathbf{a}_{\text{RIS},x}(\psi_{\mathbf{G}},\varphi_{\mathbf{G}}) \otimes \mathbf{a}_{\text{RIS},y}(\psi_{\mathbf{G}},\varphi_{\mathbf{G}}) \in \mathbb{C}^{N\times 1}$ is the RIS array steering vector where $\mathbf{a}_{\text{RIS},x}(\psi_{\mathbf{G}},\varphi_{\mathbf{G}}) = [1,\cdots,e^{-j\pi(N_{x}-1)\cos\psi_{\mathbf{G}}\sin\varphi_{\mathbf{G}}}]^{\text{T}}$ and $\mathbf{a}_{\text{RIS},y}(\psi_{\mathbf{G}},\varphi_{\mathbf{G}}) = [1,\cdots,e^{-j\pi(N_{y}-1)\sin\psi_{\mathbf{G}}\sin\varphi_{\mathbf{G}}}]^{\text{T}}$.

Second, the channel vector from the $k$-th IoT device to the RIS $\mathbf{u}_{k}$ is expressed as
\begin{equation}
    \mathbf{u}_{k} = \sqrt{\beta_{k}} \Big( \sqrt{\frac{\kappa_{k}}{\kappa_{k} + 1}}\mathbf{u}_{k}^{\textup{LoS}} + \sqrt{\frac{1}{\kappa_{k} + 1}}\mathbf{u}_{k}^{\textup{NLoS}} \Big), \label{3}
\end{equation}
where $\mathbf{u}_{k}^{\textup{LoS}} = \mathbf{a}_{\text{RIS}}(\psi_{k},\varphi_{k})$ is the LoS component with $\psi_{k}$ and $\varphi_{k}$ being the azimuth and elevation of the AoAs at the RIS, respectively, $\mathbf{u}_{k}^{\textup{NLoS}}\sim \mathcal{CN}(\mathbf{0},\mathbf{I}_{N})$ is the NLoS component, $\beta_{k}$ is the path loss, and $\kappa_{k} (\geq 0)$ is the Rician factor.
The channel vector from the $k$-th IoT device to BS $\mathbf{d}_{k}$ is modeled similarly with $\mathbf{u}_{k}$.
In general, RIS is deployed at the position where the LoS links with the BS and IoT devices are guaranteed, so $\mathbf{G}$ and $\mathbf{u}_{k}$ are dominated by the LoS paths.

\subsection{Uplink Channel Estimation}
In this subsection, we propose the uplink channel estimation process in time-division duplexing (TDD)-based RIS-aided IoT networks\footnote{In the 5G systems and beyond, TDD will be a competitive duplexing option due to the improved spectrum efficiency, better adaptation quality to asymmetric uplink/downlink traffics, low transceiver cost, and better support of the massive MIMO\cite{WJKIM2018}.}.
To enjoy the full potential of RIS, the BS needs to acquire not only the conventional direct channel $\mathbf{d}_{k}$ but also the channels reflected by the RIS (i.e., $\mathbf{G}$ and $\mathbf{u}_{k}$). 
When compared to the direct channel estimation, the estimation of RIS-aided channels is far more difficult since RIS contains a large number of reflecting elements.
Recently, an RIS architecture exploiting active reflecting elements that can reflect and also receive the signal has gained much attention to address the issue \cite{WuTowards,gong2020toward}.
Using the pilot signals received at the active reflecting elements, the BS can directly estimate the RIS-aided channel components.
However, due to the practical limitations such as the implementation cost of RF chains, hardware complexity, and power consumption, only a few active reflecting elements can be employed.
That is, the BS can acquire only partial information of the RIS-aided channels from the active reflecting elements.

\begin{figure}[!t]
\centering
\includegraphics[width=0.62\columnwidth]{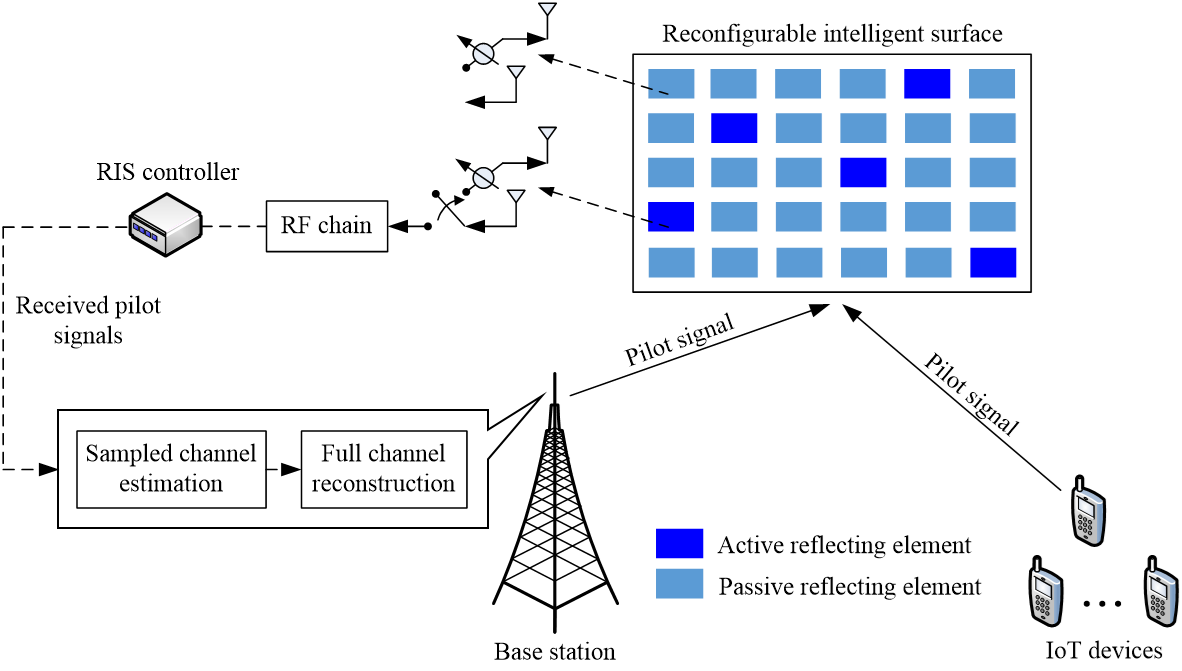}
\caption{Illustration of uplink RIS-aided channel estimation consisting of two major procedures: 1) estimation of the partial channel information observed at the active reflecting elements and 2) reconstruction of the full RIS-aided channel via the LRMC algorithm.}\label{fig2}
\end{figure}

To estimate the whole channel from the partial information, we exploit the property that the RIS-aided channels are dominated by LoS paths.
Since there is only one propagation path, the channel matrix can be readily modeled as a low-rank matrix.  
When the channel matrix has a low-rank property, BS can reconstruct the full channel matrix from the partial observations using the LRMC algorithm \cite{nguyen2019low}.
Once the full RIS-aided channel information is obtained, BS can perform the uplink power allocation and the RIS phase shift control.

\subsubsection{Sampled Channel Estimation}
We assume that the RIS consists of $\bar{N}$ active reflecting elements where the index set\footnote{Note that $\Omega \subseteq \lbrace 1, \cdots, N_{x}\rbrace \times \lbrace 1,\cdots,N_{y}\rbrace$. For example, $\Omega = \lbrace (1,2),(4,3),(5,6)\rbrace$ when $N_{x}=N_{y}=8$ and $\bar{N}=3$.} of active reflecting elements is $\Omega$ (see Fig. 2).
The uplink channel estimation process consists of $(M + K)$ time slots where the $m$-th BS antenna transmits downlink pilot signal to the RIS at $m$-th time slot ($m = 1,\cdots,M$) and then the $k$-th IoT device transmits uplink pilot signal to the RIS at $(M+k)$-th time slot ($k = 1,\cdots,K$). 

Let $x_{m}^{\text{BS}}$ be the downlink pilot signal of $m$-th BS antenna at $m$-th time slot and $x_{k}^{\text{D}}$ be the downlink pilot signal of $k$-th IoT device at $(M+k)$-th time slot.
Then the received signal matrices $\mathbf{Y}_{m}^{\text{BS}}$,$\mathbf{Y}_{k}^{\text{D}}\in\mathbb{C}^{N_{x}\times N_{y}}$ of RIS from the $m$-th BS antenna and the $k$-th IoT device are 
\begin{align}
    \mathbf{Y}_{m}^{\text{BS}}
    &= (x_{m}^{\text{BS}})^{*}\mathbf{P}_{\Omega}(\mathbf{G}_{m}) + \mathbf{N}_{m}^{\text{BS}}, \quad m=1,\cdots, M, \label{4}\\
    \mathbf{Y}_{k}^{\text{D}}
    &= x_{k}^{\text{D}}\mathbf{P}_{\Omega}(\mathbf{U}_{k}) + \mathbf{N}_{k}^{\text{D}}, \quad k=1,\cdots, K, \label{5}
\end{align}
where $\mathbf{G}_{m} \in \mathbb{C}^{N_{x}\times N_{y}}$ is the channel matrix from the $m$-th BS antenna to the RIS such that $\text{vec}(\mathbf{G}_{m})$ is the $m$-th row vector of $\mathbf{G}$, $\mathbf{U}_{k}\in\mathbb{C}^{N_{x}\times N_{y}}$ is the channel matrix from the $k$-th IoT device to the RIS such that $\text{vec}(\mathbf{U}_{k}) = \mathbf{u}_{k}$, and $\mathbf{N}_{m}^{\text{BS}}$ and $\mathbf{N}_{k}^{\text{D}}$ are the additive Gaussian noises.
Also, $\mathbf{P}_{\Omega}(\mathbf{G}_{m})$ and $\mathbf{P}_{\Omega}(\mathbf{U}_{k})$ are the sampled matrices\footnote{$\mathbf{P}_{\Omega}$ is the sampling operator such that $[\mathbf{P}_{\Omega}(\mathbf{A})]_{x,y} = [\mathbf{A}]_{x,y}$ if $(x,y) \in \Omega$ and otherwise zero for a matrix $\mathbf{A}$.} of $\mathbf{G}_{m}$ and $\mathbf{U}_{k}$.
Note that the sampling operator is used since only active reflecting elements can receive the pilot signals. 
Then the BS can easily acquire the sampled channel matrices $\mathbf{P}_{\Omega}(\mathbf{G}_{m})$ and $\mathbf{P}_{\Omega}(\mathbf{U}_{k})$ using the conventional least squares (LS) and minimum mean square error (MMSE) estimation techniques. 

\subsubsection{Full Channel Reconstruction}
After the sampled channel estimation, the BS reconstructs the full channel matrices $\lbrace\mathbf{G}_{m}\rbrace_{m=1}^{M}$ and $\lbrace\mathbf{U}_{k}\rbrace_{k=1}^{K}$ from $\lbrace\mathbf{P}_{\Omega}(\mathbf{G}_{m})\rbrace_{m=1}^{M}$ and $\lbrace\mathbf{P}_{\Omega}(\mathbf{U}_{k})\rbrace_{k=1}^{K}$ via the LRMC algorithm \cite{nguyen2019low}. 

To be specific, $\mathbf{G}_{m}$ can be reconstructed by solving the rank minimization problem:
\begin{subequations} \label{6}
\begin{align}
\min_{\mathbf{X}\in\mathbb{C}^{N_{x}\times N_{y}}} \,
&\text{rank}(\mathbf{X}) \label{6a} \\
\mbox{s.t.} \quad\,\,
&\mathbf{P}_{\Omega}(\mathbf{X}) = \mathbf{P}_{\Omega}(\mathbf{G}_{m}).  \label{6b}
\end{align} 
\end{subequations}
The solution $\mathbf{X}^{*}$ of \eqref{6} is the estimate of $\mathbf{G}_{m}$.
Since the rank minimization problem is NP-hard, this problem is computationally intractable.
To deal with the problem, we replace the non-convex objective function with its convex surrogate.
The nuclear norm $\|\mathbf{X}\|_{*}$, the sum of the singular values of $\mathbf{X}$, has been widely used as a convex surrogate of $\text{rank}(\mathbf{X})$\footnote{It has been shown that the nuclear norm is the convex envelope of rank function on the set $\lbrace \mathbf{X}: \|\mathbf{X}\| \leq 1\rbrace$\cite{nguyen2019low}.}:
\begin{subequations} \label{7}
\begin{align}
\min_{\mathbf{X}\in\mathbb{C}^{N_{x}\times N_{y}}}  \,\,\, 
&\|\mathbf{X}\|_{*}  \label{7a}\\
\mbox{s.t.} \quad\,\,\,\,
&\mathbf{P}_{\Omega}(\mathbf{X})= \mathbf{P}_{\Omega}(\mathbf{G}_{m}). \label{7b}  
\end{align}
\end{subequations}

The nuclear norm minimization problem \label{7} can also be recast as a semidefinite programming (SDP) \cite{nguyen2019low}:
\begin{subequations} \label{8}
\begin{align}
\min_{\mathbf{Z}} \,\,\,
&\text{tr}(\mathbf{Z}) \label{8a} \\
\mbox{s.t.} \,\,\,
&\text{tr}(\mathbf{A}_{x,y}^{\text{H}}\mathbf{Z})  = [\mathbf{G}_{m}]_{x,y}, \quad (x,y)\in\Omega, \label{8b} \\
&\mathbf{Z}\succeq \mathbf{0}, \label{8c}
\end{align}
\end{subequations}
where  $\mathbf{Z} = {\small \begin{bmatrix} \mathbf{Z}_{1} & \mathbf{X}\\ \mathbf{X}^{\text{H}} & \mathbf{Z}_{2} \end{bmatrix}} \in \mathbb{C}^{(N_{x}+N_{y}) \times (N_{x}+N_{y})}$ for the Hermitian matrices $\mathbf{Z}_{1} \in \mathbb{C}^{N_{x} \times N_{x}}$ and $\mathbf{Z}_{2} \in \mathbb{C}^{N_{y} \times N_{y}}$, $\mathbf{A}_{x,y} = \mathbf{e}_{x}\mathbf{e}_{y+N_{x}}^{\text{T}}$ is the linear sampling matrix, and $\mathbf{e}_{x}$ is the $x$-th column vector of $\mathbf{I}_{N_{x}+N_{y}}$.
Since \eqref{8} is a convex problem, the solutions $\mathbf{Z}^{*}$ and $\mathbf{X}^{*}$ of \eqref{8} can be obtained using the convex optimization solvers (e.g., CVX \cite{GrantCVX}).
Finally, $\mathbf{G}_{m}$ is obtained by $\mathbf{G}_{m} = \mathbf{X}^{*}$.

Similarly, $\mathbf{U}_{k}$ can be recovered from $\mathbf{P}_{\Omega}(\mathbf{U}_{k})$ by solving the low-rank matrix completion problem.
Finally, the whole RIS-aided channels $\mathbf{G}$ and $\mathbf{u}_{k}$ are obtained as $\mathbf{G}=[\text{vec}(\mathbf{G}_{1})\cdots\text{vec}(\mathbf{G}_{M})]^{\text{T}}$ and $\mathbf{u}_{k}=\text{vec}(\mathbf{U}_{k})$.

\subsection{Uplink Data Transmission}
After the channel estimation, BS performs the uplink power allocation and the RIS phase shift control using the acquired channel information.
Then the BS sends the uplink transmit power indicator $p_{k}$ and the RIS phase shift vector $\boldsymbol{\theta}$ to the $k$-th IoT device and RIS, respectively.
Finally, each IoT device transmits the data to BS through the uplink channel.

Let ${x}_{k}= \sqrt{{p}_{k}}{s}_{k}$ be the data signal of the $k$-th IoT device where ${s}_{k}$ and ${p}_{k} (\geq 0)$ are the normalized data symbol and the transmit power of the $k$-th IoT device, respectively.
Then, the received signal at BS from the $k$-th IoT device $y_{k}$ is
\begin{align}
y_{k} 
&= \mathbf{w}_{k}^{\text{H}}\bigg(\mathbf{h}_{k}{x}_{k} + \sum_{j\neq k}^{K} \mathbf{h}_{j}{x}_{j} + \mathbf{n}_{k} \bigg) \label{9}\\
&= \sqrt{{p}_{k}}\mathbf{w}_{k}^{\text{H}}(\mathbf{d}_{k}+\mathbf{G}\mathbf{H}_{k}\boldsymbol{\theta}){s}_{k} + \sum_{j\neq k}^{K} \sqrt{{p}_{j}}\mathbf{w}_{k}^{\text{H}}(\mathbf{d}_{j}+\mathbf{G}\mathbf{H}_{j}\boldsymbol{\theta}){s}_{j} + \mathbf{w}_{k}^{\text{H}}\mathbf{n}_{k}, \label{10}  
\end{align}
where $\mathbf{w}_{k} \in \mathbb{C}^{M \times 1}$ is the normalized BS beamforming vector for the $k$-th IoT device, i.e., $\|\mathbf{w}_{k}\|=1$, and $\mathbf{n}_{k} \sim \mathcal{CN}(\mathbf{0},\sigma_{k}^{2}\mathbf{I})$ is the additive Gaussian noise.
In this setting, the uplink achievable rate $R_{k}$ of the $k$-th IoT device is given by
\begin{equation}
R_{k} = \text{log}_{2}\Bigg(1 + \frac {p_{k}|\mathbf{w}_{k}^{\text{H}}(\mathbf{d}_{k}+\mathbf{G}\mathbf{H}_{k}\boldsymbol{\theta})|^{2}} 
{\sum_{j\neq k}^{K}p_{j}|\mathbf{w}_{k}^{\text{H}}(\mathbf{d}_{j}+\mathbf{G}\mathbf{H}_{j}\boldsymbol{\theta})|^{2}+\sigma_{k}^{2}}\Bigg).  \label{11}
\end{equation}

\subsection{Uplink Power Minimization Problem Formulation}
The uplink power minimization problem to optimize the RIS phase shift vector $\boldsymbol{\theta}$, the BS beamforming matrix $\mathbf{W} \!\!=\!\! [\mathbf{w}_{\!1}\!  \cdots\! \mathbf{w}_{\!K}\!]$, and the device power vector $\mathbf{p} \!\!=\!\! [p_{1}\! \cdots\!  p_{\!K}\!]$ is formulated as
\begin{subequations} \label{12}
\begin{align}
\mathcal{P}_{1}: \min_{ \boldsymbol{\theta}, \mathbf{W}, \mathbf{p}} \,\,
&\sum_{k=1}^{K}p_{k} \label{12a} \\
\mbox{s.t.} \,\,\,\,
& \text{log}_{2}\Bigg(1 + \frac {p_{k}|\mathbf{w}_{k}^{\text{H}}(\mathbf{d}_{k}+\mathbf{G}\mathbf{H}_{k}\boldsymbol{\theta})|^{2}} 
{\sum_{j\neq k}^{K}p_{j}|\mathbf{w}_{k}^{\text{H}}(\mathbf{d}_{j}+\mathbf{G}\mathbf{H}_{j}\boldsymbol{\theta})|^{2}+\sigma_{k}^{2}}\Bigg) \geq R_{k}^{\textup{min}}, \quad \forall k \in \mathcal{K}, \label{12b} \\ 
& |\theta_{n}|=1, \quad  \forall  n \in \mathcal{N}, \label{12c}\\
& \|\mathbf{w}_{k}\|=1, \quad \forall  k \in \mathcal{K},   \label{12d}\\
& 0 \leq p_{k} \leq p_{k}^{\text{max}}, \quad  \forall  k \in \mathcal{K},   \label{5e}
\end{align} 
\end{subequations}
where $\mathcal{K}$ and $\mathcal{N}$ are the sets of IoT devices and RIS reflecting elements and $R_{k}^{\textup{min}}$ and $p_{k}^{\text{max}}$ are the rate requirement and the maximum transmit power of the $k$-th IoT device, respectively.
Note that \eqref{12c} is the unit-modulus constraint of the RIS phase shift and \eqref{12d} is the unit-norm constraint of the BS beamforming vector.
Due to the nonconvexity of the unit-modulus and unit-norm constraints, $\mathcal{P}_{1}$ is a non-convex problem.
This, together with the quadratic fractional and coupled structure of the rate function in \eqref{12b}, makes $\mathcal{P}_{1}$ very difficult to solve.

\section{Riemannian Conjugate Gradient-based joint optimization Algorithm}
The primal goal of the proposed RCG-JO technique is to find out the RIS phase shifts and the BS beamforming vectors minimizing the uplink transmit power of RIS-aided IoT networks.
As mentioned, main obstacles in solving the uplink power minimization problem are the nonconvex unit-modulus constraint of the RIS phase shift and unit-norm constraint of the BS beamforming vector.
To handle these issues, we exploit the smooth product Riemannian manifold structure of the sets of unit-modulus phase shifts and unit-norm beamforming vectors.
Since the product of two manifolds is also a Riemannian manifold with smooth structure, we can readily convert the uplink power minimization problem $\mathcal{P}_{1}$ to an unconstrained problem on the product Riemannian manifold.
After that, by using the differential geometry tools, we can design the gradient descent algorithm on the Riemannian manifold and use it to obtain the optimal RIS phase shifts and the BS beamforming vectors minimizing the uplink transmit power of the RIS-aided IoT network.

In a nutshell, the proposed RCG-JO algorithm consists of two major steps: 
1) Lagrangian relaxation to move the complicated rate constraint to the objective function and 2) alternating optimization of $\mathbf{p}$ and $(\boldsymbol{\theta}, \mathbf{W})$ to minimize the modified objective function on the product Riemannian manifold. 
To be specific, in the alternating optimization step, we first fix the device power $\mathbf{p}$ and then jointly optimize the phase shift vector $\boldsymbol{\theta}$ and the BS beamforming matrix $\mathbf{W}$ using the RCG method. 
Once the optimal $\boldsymbol{\theta}$ and $\mathbf{W}$ are obtained, the optimization problem of $\mathbf{p}$ is formulated as a linear programming (LP) problem where the optimal solution can be easily obtained using the convex optimization technique.
We repeat these procedures until the objective function $P_{\textup{total}} = \sum_{k=1}^{K}p_{k}$ converges (see Algorithm 1). 

\subsection{Notions on Riemannian Manifolds}
In this subsection, we briefly introduce properties and operators of the optimization on Riemannian manifold.
Roughly speaking, a smooth manifold is a generalization of the Euclidean space on which the notion of differentiability exists \cite{AbsilOptimization}.
The tangent space $\mathcal{T}_{\mathbf{X}}\mathcal{Y}$ at a point $\mathbf{X}$ of a manifold $\mathcal{Y}$ is the set of the tangent vectors of all the curves at $\mathbf{X}$, where the curve $\gamma$ of $\mathcal{Y}$ is a mapping from $\mathbb{C}$ to $\mathcal{Y}$.
Put it formally, for a given point $\mathbf{X} \in \mathcal{Y}$, the tangent space of $\mathcal{Y}$ at $\mathbf{X}$ is defined as $\mathcal{T}_{\mathbf{X}}\mathcal{Y} = \{\gamma^{\prime}(0): \gamma\;\text{is a curve in}\;\mathcal{Y}, \gamma(0) = \mathbf{X} \}$ (see Fig. 3(a)).
A manifold $\mathcal{Y}$ together with a smoothly varying inner product $g = \langle\cdot,\cdot\rangle: \mathcal{T}_{\mathbf{X}}\mathcal{Y} \times \mathcal{T}_{\mathbf{X}}\mathcal{Y} \rightarrow \mathbb{C}$ on the tangent space $\mathcal{T}_{\mathbf{X}}\mathcal{Y}$ forms a smooth Riemannian manifold, denoted as $(\mathcal{Y},g)$ where $g$ is termed as the Riemannian metric.


\begin{algorithm}[t]
\begin{algorithmic}[0]
\item[\textbf{Input:}] Number of iterations $T$, rate requirement $R_{k}^{\textup{min}}$, maximum transmit power $p_{k}^{\textup{max}}$.
\item[\textbf{Output:}] Uplink transmit power $\mathbf{p}$, RIS phase shifts $\boldsymbol{\theta}$, and BS beamforming matrix $\mathbf{W}$.
\item[\textbf{Initialization:}] $t = 1$, $\mathbf{p}_{t} = \mathbf{p}_{\text{ini}}$, $\boldsymbol{\theta}_{t} = \boldsymbol{\theta}_{\text{ini}}$, $\mathbf{W}_{t} = \mathbf{W}_{\text{ini}}$.
\algrule
\While{$P_{\textup{total}}$ \textup{does not converge}} 
\State Optimize $\boldsymbol{\theta}_{t}$ and $\mathbf{W}_{t}$ simultaneously using Algorithm 2 when $\mathbf{p}_{t}$ is fixed
\State Optimize $\mathbf{p}_{t}$ by solving an LP problem when $\boldsymbol{\theta}_{t}$ and $\mathbf{W}_{t}$ are fixed 
\State $t = t+1$
\EndWhile
\caption{Riemannian conjugate gradient-based joint optimization algorithm}
\end{algorithmic}
\end{algorithm}

When the Riemannian manifold $(\mathcal{Y},g)$ is the cartesian product of two Riemannian manifolds $(\mathcal{Y}_{1},g_{1})$ and $(\mathcal{Y}_{2},g_{2})$, then the Riemannian metric $g$ is defined as $g = g_{1} + g_{2}$ \cite{becigneul2018riemannian}.
In the following lemma, we show that the tangent space on the product Riemannian manifold is the direct sum of the tangent spaces on each Riemannian manifold.

\begin{lemma}
For the point $\mathbf{X} = \mathbf{X}_{1} \oplus \mathbf{X}_{2}$ where $\mathbf{X}_{1} \in \mathcal{Y}_{1}, \,\mathbf{X}_{2} \in \mathcal{Y}_{2}$, the tangent space at $\mathbf{X}$ of the product Riemannian manifold $\mathcal{Y}=\mathcal{Y}_{1}\times\mathcal{Y}_{2}$ is given by \cite{becigneul2018riemannian}
\begin{equation}
\mathcal{T}_{\mathbf{X}}\mathcal{Y}
= \mathcal{T}_{\mathbf{X}_{1}}\mathcal{Y}_{1}  \oplus \mathcal{T}_{\mathbf{X}_{2}}\mathcal{Y}_{2}. \label{13}
\end{equation}
\end{lemma}

Given a smooth objective function $f$ on the Riemannian manifold $\mathcal{Y}$, minimization of $f$ requires the notion of gradient at every $\mathbf{X} \in \mathcal{Y}$.
To be specific, the Riemannian gradient of $f$ at $\mathbf{X}$, denoted by $\textup{grad}_{\mathcal{Y}}f(\mathbf{X})$, is defined as a unique vector on $\mathcal{T}_{\mathbf{X}}\mathcal{Y}$ that yields the steepest-descent of $f$.
Put it formally, $\textup{grad}_{\mathcal{Y}}f(\mathbf{X})$ is obtained by projecting $\nabla_{\mathbf{X}}f(\mathbf{X})$ onto $\mathcal{T}_{\mathbf{X}}\mathcal{Y}$ (see Fig. 3(b)).

Using Lemma 1, the Riemannian gradient $\textup{grad}_{\mathcal{Y}}f(\mathbf{X})$ is given by
\begin{equation}
\textup{grad}_{\mathcal{Y}}f(\mathbf{X}) = \textup{grad}_{\mathcal{Y}_{1}}f_{1}(\mathbf{X}_{1}) \oplus \textup{grad}_{\mathcal{Y}_{2}}f_{2}(\mathbf{X}_{2}), \label{14}
\end{equation}
where $\textup{grad}_{\mathcal{Y}_{i}}f_{i}(\mathbf{X}_{i}) \in \mathcal{T}_{\mathbf{X}_{i}}\mathcal{Y}_{i}$ is the Riemannian gradient of $f_{i}$ at $\mathbf{X}_{i} \in \mathcal{Y}_{i}$ for $i = 1,2$.

\begin{figure}[!t]
\centering
\includegraphics[width=0.85\columnwidth]{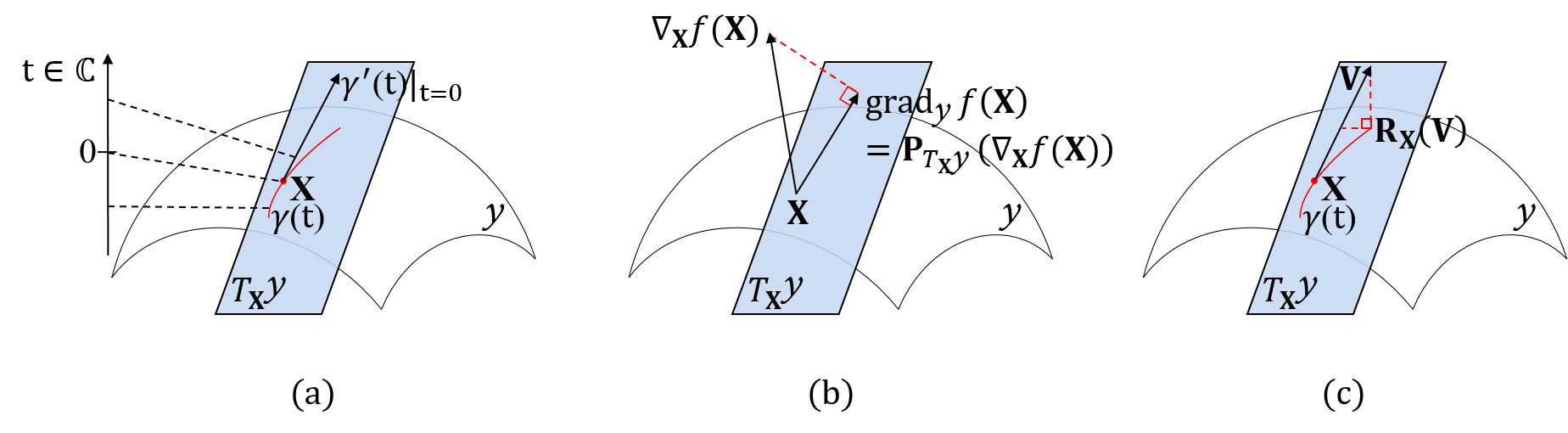}
\vspace{-0.3cm}
\caption{Illustration of (a) the tangent space $\mathcal{T}_{\mathbf{X}}\mathcal{Y}$, (b) the Riemannian gradient $\textup{grad}_{\mathcal{Y}}f(\mathbf{X})$, and (c) the retraction operator $\textbf{\textup{R}}_{\mathbf{X}}(\mathbf{V})$ at the point $\mathbf{X}$ in the Riemannian manifold $\mathcal{Y}$. 
Note that Euclidean gradient $\nabla_{\mathbf{X}}f(\mathbf{X})$ is a direction for which the objective function is reduced in $\mathbb{C}^{P}$ whereas the Riemannian gradient $\text{grad}_{\mathcal{Y}}f(\mathbf{X})$ is a direction for which the objective function is reduced in $\mathcal{T}_{\mathbf{X}}\mathcal{Y}$.
}
\end{figure}

\begin{defn}
An orthogonal projection onto the tangent space $\mathcal{T}_{\mathbf{X}}\mathcal{Y}$ is a mapping $\textbf{\textup{P}}_{\mathcal{T}_{\mathbf{X}}\mathcal{Y}}: \mathbb{C}^{a\times b}\to \mathcal{T}_{\mathbf{X}}\mathcal{Y}$ such that for a given $\mathbf{U}\in\mathbb{C}^{a\times b}$, $\langle \mathbf{U} - \textbf{\textup{P}}_{\mathcal{T}_{\mathbf{X}}\mathcal{Y}}(\mathbf{U}),\mathbf{V} \rangle = 0$ for all $\mathbf{V} \in \mathcal{T}_{\mathbf{X}}\mathcal{Y}$. In particular, the orthogonal projection onto the tangent space of product manifold $\mathcal{Y}=\mathcal{Y}_{1}\times\mathcal{Y}_{2}$ is
\begin{equation}
\textbf{\textup{P}}_{\mathcal{T}_{\mathbf{X}}\mathcal{Y}}(\mathbf{U}) = \textbf{\textup{P}}_{\mathcal{T}_{\mathbf{X}_{1}}\mathcal{Y}_{1}}(\mathbf{U}_{1}) \oplus \textbf{\textup{P}}_{\mathcal{T}_{\mathbf{X}_{2}}\mathcal{Y}_{2}}(\mathbf{U}_{2}), \label{15}
\end{equation}
where $\textbf{\textup{P}}_{\mathcal{T}_{\mathbf{X}_{i}}\mathcal{Y}_{i}}(\mathbf{U}_{i})$ is the projection of $\mathbf{U}_{i}$ onto $\mathcal{T}_{\mathbf{X}_{i}}\mathcal{Y}_{i}$ for $i=1,2$.
\end{defn}

In order to express the concept of moving in the direction of a tangent space while staying on the manifold, we need an operation called \emph{retraction}.
As illustrated in Fig. 3(c), the retraction is a mapping from $\mathcal{T}_{\mathbf{X}}\mathcal{Y}$ to $\mathcal{Y}$ that preserves the gradient at $\mathbf{X}$\cite{AbsilOptimization}.

\begin{defn}
The retraction $\textbf{\textup{R}}_{\mathbf{X}}(\mathbf{V})$ of a matrix $\mathbf{V} \in \mathcal{T}_{\mathbf{X}}\mathcal{Y}$ onto $\mathcal{Y}$ is defined as
\begin{equation}
\textbf{\textup{R}}_{\mathbf{X}}(\mathbf{V}) = \arg\min_{\mathbf{Z} \in \mathcal{Y}} \|\mathbf{X} + \mathbf{V} - \mathbf{Z}\|_{\textup{F}}. \label{16}
\end{equation}
In particular, the retraction onto the product manifold $\mathcal{Y}=\mathcal{Y}_{1}\times\mathcal{Y}_{2}$ is
\begin{equation}
\textbf{\textup{R}}_{\mathbf{X}}(\mathbf{V}) 
= \textbf{\textup{R}}_{\mathbf{X}_{1}}(\mathbf{V}_{1}) \oplus  \textbf{\textup{R}}_{\mathbf{X}_{2}}(\mathbf{V}_{2}), \label{17}
\end{equation}
where $\textbf{\textup{R}}_{\mathbf{X}_{i}}(\mathbf{V}_{i})$ is the retraction of $\mathbf{V}_{i} \in \mathcal{T}_{\mathbf{X}_{i}}\mathcal{Y}_{i}$ onto $\mathcal{Y}_{i}$ for $i=1,2$.
\end{defn}

\subsection{Joint RIS Phase Shift and BS Beamforming Optimization on Product Manifold}
In the first step, we fix the device power and then jointly optimize the RIS phase shifts and BS beamforming vectors using the RCG method.

When the device power is given, $\mathcal{P}_{1}$ is reduced to
\begin{subequations} \label{18}
\begin{align}
\mathcal{P}_{2}: \text{Find} \,\,\,\, &(\boldsymbol{\theta}, \mathbf{W})  \label{18a} \\
\mbox{s.t.} \quad
& \frac{p_{k}}{2^{R_{k}^{\textup{min}}} -1 } A_{k,k}(\boldsymbol{\theta},\mathbf{w}_{k}) - \sum_{j \neq k}^{K} p_{j}\, A_{j,k}(\boldsymbol{\theta},\mathbf{w}_{k}) \geq \sigma_{k}^{2}, \quad \forall  k \in \mathcal{K},  \label{18b}\\
& |\theta_{n}| = 1, \quad \forall  n \in \mathcal{N}, \label{18c}\\
& \|\mathbf{w}_{k}\| = 1, \quad \forall  k \in \mathcal{K},   \label{18d}
\end{align}
\end{subequations}
where $A_{j,k}(\boldsymbol{\theta},\mathbf{w}_{k}) = \|\mathbf{w}_{k}^{\text{H}}(\mathbf{d}_{j}+\mathbf{G}\mathbf{H}_{j}\boldsymbol{\theta})\|^{2}$ for $j,k \in \mathcal{K}$.
Since the rate expression in \eqref{18b} is a joint quadratic function of  $\boldsymbol{\theta}$ and $\mathbf{W}$, \eqref{18b} is a nonconvex constraint. 
To handle this issue, we use the Lagrangian relaxation to move the complicated rate constraints to the objective function.
To be specific, the modified objective function is given by
\begin{equation}
L(\boldsymbol{\theta}, \mathbf{W}, \boldsymbol{\lambda}) = \sum_{k=1}^{K} \lambda_{k} \bigg( - \frac{p_{k}}{2^{R_{k}^{\textup{min}}} -1} A_{k,k}(\boldsymbol{\theta},\mathbf{w}_{k}) + \sum_{j \neq k}^{K} p_{j} A_{j,k}(\boldsymbol{\theta},\mathbf{w}_{k}) + \sigma_{k}^{2} \bigg), \label{19}
\end{equation}
where $\boldsymbol{\lambda} = [\lambda_{1},\cdots,\lambda_{K}]$ is the Lagrangian multiplier obtained by solving the corresponding dual problem.
Using $L(\boldsymbol{\theta}, \mathbf{W}, \boldsymbol{\lambda})$, $\mathcal{P}_{2}$ is relaxed into
\begin{subequations}\label{20}
\begin{align}
\mathcal{P}_{3}: \min_{\boldsymbol{\theta}, \mathbf{W}, \boldsymbol{\lambda}} \,\,\, &\sum_{k=1}^{K} \lambda_{k} \bigg( - \frac{p_{k}}{2^{R_{k}^{\textup{min}}} - 1} A_{k,k}(\boldsymbol{\theta},\mathbf{w}_{k}) + \sum_{j \neq k}^{K} p_{j} A_{j,k}(\boldsymbol{\theta},\mathbf{w}_{k}) + \sigma_{k}^{2}\bigg) \label{20a}\\ 
\mbox{s.t.} \,\,\,\,\,
& |\theta_{n}|=1,  \quad \forall n \in \mathcal{N},  \label{20b}\\
&\|\mathbf{w}_{k}\|=1,  \quad \forall k \in \mathcal{K}, \label{20c}\\
& \lambda_{k} \geq 0, \quad \forall k \in \mathcal{K}. \label{20d}
\end{align} 
\end{subequations}

The relaxed problem $\mathcal{P}_{3}$ looks simpler than $\mathcal{P}_{2}$, but it is still nonconvex and difficult to solve since the objective function of $\mathcal{P}_{3}$ is a joint quadratic function of $\boldsymbol{\theta}$ and $\mathbf{W}$. 
Additionally, we need to deal with the unit-modulus constraints \eqref{20b} and the unit-norm constraints \eqref{20c}.  
To handle the problem, we jointly optimize $\boldsymbol{\theta}$ and $\mathbf{W}$ on the product manifold of the unit-modulus phase shifts and the unit-norm beamforming vectors using the RCG method.
Once we obtain $\boldsymbol{\theta}$ and $\mathbf{W}$, we update the Lagrangian multiplier $\boldsymbol{\lambda}$.
We repeat these procedures until $\boldsymbol{\theta}$ and $\mathbf{W}$ converge.

\subsubsection{Joint RIS Phase Shift and BS Beamforming Optimization}
For a given $\boldsymbol{\lambda}$, $\mathcal{P}_{3}$ is reduced to the unconstrained problem on the product manifold:
\begin{equation}
\mathcal{P}_{(\boldsymbol{\theta},\mathbf{W})}: \min_{(\boldsymbol{\theta},\mathbf{W}) \in \mathcal{M}_{\boldsymbol{\theta}} \times \mathcal{M}_{\mathbf{W}}} L(\boldsymbol{\theta}, \mathbf{W}), \label{21}
\end{equation} 
where $\mathcal{M}_{\boldsymbol{\theta}}$ and $\mathcal{M}_{\mathbf{W}}$ are the complex circle manifold and complex oblique manifold given by
\begin{align}
\mathcal{M}_{\boldsymbol{\theta}} &= \{\boldsymbol{\theta}\in\mathbb{C}^{N \times 1}: |\theta_{n}| = 1,\; \forall n \in \mathcal{N}\} \label{22}\\
\mathcal{M}_{\mathbf{W}} &= \{\mathbf{W} \in \mathbb{C}^{M \times K}: \text{ddiag}(\mathbf{W}^{\text{H}}\mathbf{W}) = \mathbf{I}_{K}\}, \label{23}
\end{align}
with the inner products defined as $g_{\boldsymbol{\theta}}(\mathbf{z}_{1}, \mathbf{z}_{2}) = \langle \mathbf{z}_{1}, \mathbf{z}_{2} \rangle = \Re\{\mathbf{z}_{1}^{\text{H}} \mathbf{z}_{2}\}$ and $g_{\mathbf{W}}(\mathbf{Z}_{1}, \mathbf{Z}_{2}) = \langle \mathbf{Z}_{1}, \mathbf{Z}_{2} \rangle = \Re \{\text{tr}(\mathbf{Z}_{1}^{\text{H}}\mathbf{Z}_{2})\}$, respectively.

By combining $\boldsymbol{\theta}$ and $\mathbf{W}$ into $\mathbf{\Sigma} = \text{blkdiag}(\boldsymbol{\theta}, \mathbf{W})$, $\mathcal{P}_{(\boldsymbol{\theta},\mathbf{W})}$ is re-expressed as
\begin{equation}
\mathcal{P}_{\mathbf{\Sigma}}: \min_{\mathbf{\Sigma} \in \mathcal{M}}\;\, L(\mathbf{\Sigma}),  \label{24}
\end{equation} 
where $\mathcal{M} = \mathcal{M}_{\boldsymbol{\theta}} \times \mathcal{M}_{\mathbf{W}}$ is the product manifold with the inner product defined as $g_{\mathbf{\Sigma}} = g_{\boldsymbol{\theta}} + g_{\mathbf{W}}$.
In the following lemma, we provide the tangent space of the product manifold $\mathcal{M}$.
\begin{lemma}
The tangent space $\mathcal{T}_{\mathbf{\Sigma}}\mathcal{M}$ of the product manifold $\mathcal{M}$ at the point $\mathbf{\Sigma}$ is given by
\begin{equation}
\mathcal{T}_{\mathbf{\Sigma}}\mathcal{M}
= \mathcal{T}_{\boldsymbol{\theta}}\mathcal{M}_{\boldsymbol{\theta}} \oplus \mathcal{T}_{\mathbf{W}}\mathcal{M}_{\mathbf{W}}, \label{25}
\end{equation}
where $\mathcal{T}_{\boldsymbol{\theta}}\mathcal{M}_{\boldsymbol{\theta}} = \{\mathbf{z}\in\mathbb{C}^{N \times 1} :\Re\{\mathbf{z}^{*} \odot \boldsymbol{\theta} \} = {\mathbf{0}}_{N} \}$ is the tangent space of $\mathcal{M}_{\boldsymbol{\theta}}$ at $\boldsymbol{\theta}$ and $\mathcal{T}_{\mathbf{W}}\mathcal{M}_{\mathbf{W}} = \{\mathbf{Z}\in\mathbb{C}^{M \times K} : \textup{ddiag} (\Re \{\mathbf{W}^{\textup{H}}\mathbf{Z}\}) = {\mathbf{0}}_{K} \}$ is the tangent space of $\mathcal{M}_{\mathbf{W}}$ at $\mathbf{W}$.
\end{lemma}

In order to minimize the objective function $L(\mathbf{\Sigma})$ in $\mathcal{P}_{\mathbf{\Sigma}}$, we need a Riemannian gradient which is obtained by projecting the Euclidean gradient of $L(\mathbf{\Sigma})$ at $\mathbf{\Sigma}$ onto $\mathcal{T}_{\mathbf{\Sigma}}\mathcal{M}$. 
\begin{lemma}
The orthogonal projection $\textbf{\textup{P}}_{\mathcal{T}_{\mathbf{\Sigma}}\mathcal{M}}(\bar{\mathbf{U}})$ of $\bar{\mathbf{U}}=\textup{blkdiag}(\mathbf{u},\mathbf{U})$ onto $\mathcal{T}_{\mathbf{\Sigma}} \mathcal{M}$ is
\begin{equation}
\textbf{\textup{P}}_{\mathcal{T}_{\mathbf{\Sigma}}\mathcal{M}}(\bar{\mathbf{U}}) = \textbf{\textup{P}}_{\mathcal{T}_{\boldsymbol{\theta}}\mathcal{M}_{\boldsymbol{\theta}}}(\mathbf{u}) \oplus \textbf{\textup{P}}_{\mathcal{T}_{\mathbf{W}}\mathcal{M}_{\mathbf{W}}}(\mathbf{U}), \label{26}
\end{equation}
where $\textbf{\textup{P}}_{\mathcal{T}_{\boldsymbol{\theta}}\mathcal{M}_{\boldsymbol{\theta}}}(\mathbf{u}) = \mathbf{u} - \Re \{\boldsymbol{\theta}^{*} \odot \mathbf{u}\} \odot \boldsymbol{\theta}$ is the orthogonal projection of $\mathbf{d}$ onto $\mathcal{T}_{\boldsymbol{\theta}}\mathcal{M}_{\boldsymbol{\theta}}$ and $\textbf{\textup{P}}_{\mathcal{T}_{\mathbf{W}}\mathcal{M}_{\mathbf{W}}}(\mathbf{U}) \!=\! \mathbf{U} - \mathbf{W} \, \textup{ddiag}(\Re \{\mathbf{W}^{\textup{H}} \mathbf{U}\})$ is the orthogonal projection of $\mathbf{U}$ onto $\mathcal{T}_{\mathbf{W}}\mathcal{M}_{\mathbf{W}}$ \cite{AbsilOptimization}.
\end{lemma}

To make sure that the point $\mathbf{\Sigma}$ is updated in the direction of the tangent space $\mathcal{T}_{\mathbf{\Sigma}}\mathcal{M}$ while staying on $\mathcal{M}$, a \emph{retraction} operation, a mapping from $\mathcal{T}_{\mathbf{\Sigma}}\mathcal{M}$ to $\mathcal{M}$, is needed.
\begin{lemma}
The retraction $\textbf{\textup{R}}_{\mathbf{\Sigma}}(\bar{\mathbf{V}})$ of $\bar{\mathbf{V}} = \textup{blkdiag}(\boldsymbol{\mathbf{v}}, \mathbf{V}) \in \mathcal{T}_{\mathbf{\Sigma}}\mathcal{M}$ is
\begin{equation}
\textbf{\textup{R}}_{\boldsymbol{\Sigma}}(\bar{\mathbf{V}}) = \textbf{\textup{R}}_{\boldsymbol{\theta}}(\mathbf{v}) \oplus  \textbf{\textup{R}}_{\mathbf{W}}(\mathbf{V}), \label{27}
\end{equation}
where $\textbf{\textup{R}}_{\boldsymbol{\theta}}(\mathbf{v}) = (\boldsymbol{\theta} + \mathbf{v}) \odot \frac{1}{|\boldsymbol{\theta} + \mathbf{v}|}$ is the retraction of $\mathbf{v} \in \mathcal{T}_{\boldsymbol{\theta}}\mathcal{M}_{\boldsymbol{\theta}}$ and $\textbf{\textup{R}}_{\mathbf{W}}(\mathbf{V}) = \frac{(\mathbf{W} + \mathbf{V})}{\|\textup{ddiag}((\mathbf{W} + \mathbf{V})^{\textup{H}}(\mathbf{W} + \mathbf{V}))\|_{\textup{F}}}$ is the retraction of $\mathbf{V} \in \mathcal{T}_{\mathbf{W}}\mathcal{M}_{\mathbf{W}}$ \cite{AbsilOptimization}.
\end{lemma}

To find out $\mathbf{\Sigma}$ minimizing $L(\mathbf{\Sigma})$ on the product manifold $\mathcal{M}$, we exploit the RCG method, an extension of the conjugate gradient (CG)\footnote{The update equation of the conventional CG method is $\mathbf{\Sigma}_{i+1} = \mathbf{\Sigma}_{i} + \alpha_{i} \mathbf{D}_{i}$, where $\alpha_{i}$ is the step size and $\mathbf{D}_{i}$ is the conjugate direction.
In addition, the conjugate direction is updated as $\mathbf{D}_{i} = -\nabla_{\mathbf{\Sigma}}L(\mathbf{\Sigma}) + \beta_{i}\mathbf{D}_{i-1}$, where $\nabla_{\mathbf{\Sigma}}L(\mathbf{\Sigma})$ is the Euclidean gradient of $L(\mathbf{\Sigma})$ at $\mathbf{\Sigma}$ and $\beta_{i}$ is conjugate update parameter.} method to the Riemannian manifold.
In this approach, the update equations of the conjugate direction $\mathbf{D}$ and the point $\mathbf{\Sigma}$ are given by
\begin{align}
\mathbf{D}_{i} 
&= -\textup{grad}_{\mathcal{M}}L(\mathbf{\Sigma}_{i}) + \beta_{i} \textbf{\textup{P}}_{\mathcal{T}_{\mathbf{\Sigma}_{i}}\mathcal{M}} (\mathbf{D}_{i-1})  \label{28} \\
\mathbf{\Sigma}_{i+1}
&= \textbf{\textup{R}}_{\mathbf{\Sigma}_{i}}(\alpha_{i}\mathbf{D}_{i}), \label{29}
\end{align}
where $\textup{grad}_{\mathcal{M}}L(\mathbf{\Sigma}_{i})$ is the Riemannian gradient of $L(\mathbf{\Sigma}_{i})$ at $\mathbf{\Sigma}_{i}$, $\beta_{i}$ is the Fletcher-Reeves conjugate gradient parameter, and $\alpha_{i}$ is the step size \cite{dai2000nonlinear}.

We note that the RCG method is distinct from the conventional CG method in three respects:
1) the projection of the previous conjugate direction $\mathbf{D}_{i-1}$ onto the tangent space $\mathcal{T}_{\mathbf{\Sigma}_{i}}\mathcal{M}$ is needed before performing a linear combination of $\textup{grad}_{\mathcal{M}}L(\mathbf{\Sigma})$ and $\mathbf{D}_{i-1}$ since they lie on two different spaces $\mathcal{T}_{\mathbf{\Sigma}_{i}}\mathcal{M}$ and $\mathcal{T}_{\mathbf{\Sigma}_{i-1}}\mathcal{M}$,
2) the Riemannian gradient $\textup{grad}_{\mathcal{M}}L(\mathbf{\Sigma})$ is used instead of the Euclidean gradient $\nabla_{\mathbf{\Sigma}}L(\mathbf{\Sigma})$ since we need to find out the search direction on the tangent space of $\mathcal{M}$, and
3) the retraction is required to ensure that the updated point $\mathbf{\Sigma}_{i+1}$ lies on $\mathcal{M}$.
Specifically, the Riemannian gradient $\text{grad}_{\mathcal{M}}L(\mathbf{\Sigma}) = \textbf{P}_{\mathcal{T}_{\mathbf{\Sigma}}\mathcal{M}}(\nabla_{\mathbf{\Sigma}}L(\mathbf{\Sigma}))$ is obtained by projecting the Euclidean gradient onto the tangent space. 
In the following lemma, we provide the closed-form expression of $\textup{grad}_{\mathcal{M}}L(\mathbf{\Sigma})$ of $L(\mathbf{\Sigma})$ on $\mathcal{M}$.
\begin{lemma}
The Riemannian gradient $\textup{grad}_{\mathcal{M}}L(\mathbf{\Sigma})$ of $L(\mathbf{\Sigma})$ on $\mathcal{M}$ is 
\begin{equation}
\textup{grad}_{\mathcal{M}}L(\mathbf{\Sigma}) = \textup{grad}_{\mathcal{M}_{\boldsymbol{\theta}}}L(\boldsymbol{\theta}) \oplus \textup{grad}_{\mathcal{M}_{\mathbf{W}}}L(\mathbf{W}). \label{30}
\end{equation}
Specifically, the Riemannian gradient $\textup{grad}_{\mathcal{M}_{\boldsymbol{\theta}}}L(\boldsymbol{\theta})$ of $L(\boldsymbol{\theta})$ on $\mathcal{M}_{\boldsymbol{\theta}}$ is 
\begin{equation}
\textup{grad}_{\mathcal{M}_{\boldsymbol{\theta}}}L(\boldsymbol{\theta}) = \textbf{\textup{P}}_{\mathcal{T}_{\boldsymbol{\theta}}\mathcal{M}_{\boldsymbol{\theta}}}(\nabla_{\boldsymbol{\theta}}L(\boldsymbol{\theta})) = \nabla_{\boldsymbol{\theta}}L(\boldsymbol{\theta}) - \Re \{ \boldsymbol{\theta}^{*} \odot \nabla_{\boldsymbol{\theta}}L(\boldsymbol{\theta})\} \odot \boldsymbol{\theta}, \label{31}
\end{equation}
where $\nabla_{\boldsymbol{\theta}}L(\boldsymbol{\theta})$ is the Euclidean gradient of $L(\boldsymbol{\theta})$ given by (see \eqref{19})
\begin{equation}
\nabla_{\boldsymbol{\theta}}L(\boldsymbol{\theta}) = \sum_{k=1}^{K}\lambda_{k} \bigg(- \frac{ p_{k}}{2^{R_{k}^{\textup{min}}} -1} \frac{\partial A_{k,k}(\boldsymbol{\theta},\mathbf{w}_{k})}{\partial \boldsymbol{\theta}} + \sum_{j \neq k}^{K} p_{j} \frac{\partial A_{j,k}(\boldsymbol{\theta},\mathbf{w}_{k})} {\partial \boldsymbol{\theta}} \bigg), \label{32}
\end{equation}
and 
\begin{equation}
\frac{\partial A_{j,k}(\boldsymbol{\theta},\mathbf{w}_{k})}{\partial \boldsymbol{\theta}} = \mathbf{H}_{j}^{\textup{H}}\mathbf{G}^{\textup{H}}\mathbf{w}_{k} \mathbf{w}_{k}^{\textup{H}}(\mathbf{d}_{j} + \mathbf{G}\mathbf{H}_{j}\boldsymbol{\theta}), \quad \forall j,k \in \mathcal{K}. \label{33}
\end{equation}
Also, the Riemannian gradient $\textup{grad}_{\mathcal{M}_{\mathbf{W}}}L(\mathbf{W})$ of $L(\mathbf{W})$ on $\mathcal{M}_{\mathbf{W}}$ is
\begin{equation}
\textup{grad}_{\mathcal{M}_{\mathbf{W}}}L(\mathbf{W}) = \textbf{\textup{P}}_{\mathcal{T}_{\mathbf{W}}\mathcal{M}_{\mathbf{W}}}(\nabla_{\mathbf{W}}L(\mathbf{W})) = \nabla_{\mathbf{W}}L(\mathbf{W}) - \mathbf{W} \textup{ddiag}(\Re \{\mathbf{W}^{\text{H}}  \nabla_{\mathbf{W}}L(\mathbf{W})\}), \label{34}
\end{equation}
where $\nabla_{\mathbf{W}}L(\mathbf{W}) = \left[ \frac{\partial L(\mathbf{W})}{\partial \mathbf{w}_{1}},\cdots,\frac{\partial L(\mathbf{W})}{\partial \mathbf{w}_{K}} \right]$ is the Euclidean gradient of $L(\mathbf{W})$ given by (see \eqref{19})
\begin{equation}
\frac{\partial L(\mathbf{W})}{\partial \mathbf{w}_{k}} = \lambda_{k} \bigg(- \frac{p_{k}}{2^{R_{k}^{\textup{min}}} -1} \frac{\partial A_{k,k}(\boldsymbol{\theta},\mathbf{w}_{k})}{\partial \mathbf{w}_{k}} + \sum_{j \neq k}^{K} p_{j}  \frac{\partial A_{j,k}(\boldsymbol{\theta},\mathbf{w}_{k})} {\partial \mathbf{w}_{k}} \bigg), \label{35}
\end{equation}
and 
\begin{equation}
\frac{\partial A_{j,k}(\boldsymbol{\theta},\mathbf{w}_{k})}{\partial \mathbf{w}_{k}} 
= (\mathbf{d}_{j} + \mathbf{G}\mathbf{H}_{j}\boldsymbol{\theta})(\mathbf{d}_{j} + \mathbf{G}\mathbf{H}_{j}\boldsymbol{\theta})^{\textup{H}} \mathbf{w}_{k}, \quad \forall j,k\in\mathcal{K}. \label{41}
\end{equation}
\end{lemma}

Once $\mathbf{\Sigma}$ is determined, we can obtain the phase shift vector $\boldsymbol{\theta}$ and beamforming matrix $\mathbf{W}$ by decomposing $\mathbf{\Sigma}$.
The optimization steps of $\boldsymbol{\theta}$ and $\mathbf{W}$ are summarized in Algorithm 2.

\subsubsection{Lagrangian Multiplier Update}
After updating the RIS phase shift vector $\boldsymbol{\theta}$ and the BS beamforming matrix $\mathbf{W}$, we update the Lagrangian multiplier $\boldsymbol{\lambda}$.
To be specific, $\boldsymbol{\lambda}$ is updated in the direction of maximizing the dual function $D(\boldsymbol{\lambda}) = \min_{\boldsymbol{\theta}, \mathbf{W}}\, L(\boldsymbol{\theta}, \mathbf{W}, \boldsymbol{\lambda})$ as
\begin{equation}
    \boldsymbol{\lambda} = \arg\max_{\boldsymbol{\lambda}\succeq\mathbf{0}}\,D(\boldsymbol{\lambda}). \label{37}
\end{equation}
Since $D(\boldsymbol{\lambda})$ is the optimal value of the optimization problem $\mathcal{P}_{3}$, we cannot take derivative with respect to $\boldsymbol{\lambda}$, meaning that we cannot directly use the conventional gradient ascent method in finding out the optimal value of $\boldsymbol{\lambda}$.

To address this problem, we use the subgradient method, a generalized concept of gradient method for nonsmooth convex functions \cite{boyd2003subgradient}.
In particular, the subgradient $\mathbf{g}=[g_{1},\cdots,g_{K}]^{\text{T}}$ of the dual function $D(\boldsymbol{\lambda})$ is given by
\begin{equation}
g_{k} = -\frac{p_{k}}{2^{R_{k}^{\textup{min}}} -1} A_{k,k}(\boldsymbol{\theta},\mathbf{w}_{k}) + \sum_{j \neq k}^{K} p_{j} A_{j,k}(\boldsymbol{\theta},\mathbf{w}_{k}) + \sigma_{k}^{2},\quad \forall k\in\mathcal{K}. \label{38}
\end{equation}
Using $g_{k}$, we obtain the update equation of $\boldsymbol{\lambda}$:
\begin{equation}
\boldsymbol{\lambda}_{i+1} = \max(\boldsymbol{\lambda}_{i} + \eta_{i}\mathbf{g}_{i}, \, \mathbf{0}), \label{38.1}
\end{equation}
where $\eta_{i}$ is the step size \cite{boyd2003subgradient}.

\begin{algorithm}[t]
\begin{algorithmic}[0]
\item[\textbf{Input:}] Tolerance $\epsilon$, the number of iterations $T$, $c_{1} = 0.0001$, and $c_{2} = 0.1$.
\item[\textbf{Output:}] $(\boldsymbol{\theta}, \mathbf{W})$.
\item[\textbf{Initialization:}] $i = 1$, $\mathbf{\Sigma}_{1}= \text{blkdiag}(\boldsymbol{\theta}_{1}, \mathbf{W}_{1}) \in \mathcal{M}$, $\mathbf{D}_{1} = -\textup{grad}_{\mathcal{M}}L(\mathbf{\Sigma}_{1})$, $\beta_{1}=0$.
\algrule
\While{$i \leq T$} 
\State $\textup{grad}_{\mathcal{M}}L(\mathbf{\Sigma}_{i}) = \textbf{P}_{\mathcal{T}_{\mathbf{\Sigma}_{i}}\mathcal{M}}(\nabla_{\mathbf{\Sigma}}L(\mathbf{\Sigma_{i}}))$ \Comment{Compute the Riemannian gradient}
\State $\beta_{i} = \frac{\|\textup{grad}_{\mathcal{M}}L(\mathbf{\Sigma}_{i})\|^{2}}{\| \textup{grad}_{\mathcal{M}}L(\mathbf{\Sigma}_{i-1})\|^{2}}$ \Comment{Compute the conjugate gradient parameter}
\State $\mathbf{D}_{i} = -\textup{grad}_{\mathcal{M}}L(\mathbf{\Sigma}_{i}) + \beta_{i} \textbf{\textup{P}}_{\mathcal{T}_{\mathbf{\Sigma}_{i}}\mathcal{M}}(\mathbf{D}_{i-1})$
\Comment{Update the conjugate direction}
\State Find a step size $\alpha_{i}>0$ such that \Comment{Perform a line search}
\State $\text{\quad}$ $L(\textbf{\textup{R}}_{\mathbf{\Sigma}_{i}}(\alpha_{i}\mathbf{D}_{i})) \leq L(\mathbf{\Sigma}_{i}) + c_{1}\alpha_{i}\langle \textup{grad}_{\mathcal{M}}L(\mathbf{\Sigma}_{i}),\mathbf{D}_{i}\rangle$ 
\State $\text{\quad}$ $|\langle\textup{grad}_{\mathcal{M}}L(\textbf{\textup{R}}_{\mathbf{\Sigma}_{i}}(\alpha_{i}\mathbf{D}_{i})),\textbf{\textup{P}}_{\mathcal{T}_{\textbf{\textup{R}}_{\mathbf{\Sigma}_{i}}(\alpha_{i}\mathbf{D}_{i})}\mathcal{M}} (\mathbf{D}_{i}) \rangle| \leq -c_{2}\langle \textup{grad}_{\mathcal{M}}L(\mathbf{\Sigma}_{i}),\mathbf{D}_{i}\rangle$
\State $\mathbf{\Sigma}_{i+1} = \textbf{\textup{R}}_{\mathbf{\Sigma}_{i}}(\alpha_{i} \mathbf{D}_{i})$ \Comment{Update the point}
\If{$\|\mathbf{\Sigma}_{i+1} - \mathbf{\Sigma}_{i}\|^{2} \leq \epsilon$}
\State $\boldsymbol{\theta}_{i+1} = \mathbf{\Sigma}_{i+1}(1:N, 1)$ \Comment{Extract the RIS phase shift vector}
\State $\mathbf{W}_{i+1} = \mathbf{\Sigma}_{i+1}(N+1:N+M, 2:K+1)$ \Comment{Extract the BS beamforming matrix}
\EndIf
\State $i=i+1$
\EndWhile
\caption{Optimization of $(\boldsymbol{\theta}, \mathbf{W})$ on the product manifold $\mathcal{M}$}
\end{algorithmic}
\end{algorithm}

\subsection{Uplink Transmit Power Minimization}
Once the RIS phase shift vector $\boldsymbol{\theta}$ and the BS beamforming matrix $\mathbf{W}$ are obtained, we next find out the device power vector $\mathbf{p}$ minimizing the total uplink transmit power.
When $\boldsymbol{\theta}$ and $\mathbf{W}$ are given, the original uplink transmit power minimization problem $\mathcal{P}_{1}$ is reduced to 
\begin{subequations}\label{39}
\begin{align}
\mathcal{P}_{4}: \min_{ \mathbf{p}} \,\,
&\sum_{k=1}^{K}p_{k} \label{39a}\\
\mbox{s.t.} \,\,\,
& \frac{p_{k}}{2^{R_{k}^{\textup{min}}} -1} A_{k,k}(\boldsymbol{\theta},\mathbf{w}_{k}) - \sum_{j \neq k}^{K} p_{j} A_{j,k}(\boldsymbol{\theta},\mathbf{w}_{k}) \geq \sigma_{k}^{2}, \quad \forall  k \in \mathcal{K}, \label{39b}\\
&  0 \leq p_{k} \leq p_{k}^{\text{max}}, \quad \forall  k \in \mathcal{K}. \label{46c}
\end{align} 
\end{subequations}

Since $\mathcal{P}_{4}$ is an LP optimization problem, the optimal solution can be easily obtained using the convex optimization tools (e.g., CVX\cite{GrantCVX}).

\section{Convergence and Computational Complexity Analysis of RCG-JO}
In this section, we analyze the convergence and the computational complexity of RCG-JO.

\subsection{Convergence Analysis of RCG-JO Algorithm}
Recall that the proposed RCG-JO technique consists of two major iterations: 1) inner iteration for jointly optimizing $\boldsymbol{\theta}$ and $\mathbf{W}$ on the product Riemannian manifold (Algorithm 2) and 2) outer iteration for alternately updating $(\boldsymbol{\theta},\mathbf{W})$ and $\mathbf{p}$ (Algorithm 1).
In case of Algorithm 1, due to the alternating optimization operations, it is very difficult to prove its convergence analytically so that we demonstrate its convergence from the numerical results. 
In case of Algorithm 2, we show that it converges to a local minimizer in this subsection.

We first explain the strong Wolfe conditions used to determine the step size $\alpha_{i}$.
 \begin{defn}
    A step size $\alpha_{i}$ is said to satisfy the strong Wolfe conditions, restricted to the conjugate direction $\mathbf{D}_{i}$, if the following two inequalities hold \cite{AbsilOptimization}:
    \begin{align}
        L(\textbf{\textup{R}}_{\mathbf{\Sigma}_{i}}(\alpha_{i}\mathbf{D}_{i}))
        &\leq L(\mathbf{\Sigma}_{i}) + c_{1}\alpha_{i}\langle \textup{grad}_{\mathcal{M}}L(\mathbf{\Sigma}_{i}),\mathbf{D}_{i}\rangle \label{40}\\
        |\langle\textup{grad}_{\mathcal{M}}L(\textbf{\textup{R}}_{\mathbf{\Sigma}_{i}}(\alpha_{i}\mathbf{D}_{i})),\textbf{\textup{P}}_{\mathcal{T}_{\textbf{\textup{R}}_{\mathbf{\Sigma}_{i}}(\alpha_{i}\mathbf{D}_{i})}\mathcal{M}} (\mathbf{D}_{i}) \rangle| &\leq -c_{2}\langle \textup{grad}_{\mathcal{M}}L(\mathbf{\Sigma}_{i}),\mathbf{D}_{i}\rangle. \label{41}
    \end{align}
 \end{defn}
\noindent The first condition, known as Armijo's rule, ensures that the step size decreases the cost function $L(\mathbf{\Sigma})$ sufficiently.
The second condition, known as curvature condition, ensures that the Riemannian gradient converges to zero. 
Note that since $\textup{grad}_{\mathcal{M}}L(\textbf{\textup{R}}_{\mathbf{\Sigma}_{i}}(\alpha_{i}\mathbf{D}_{i}))\in\mathcal{T}_{\textbf{\textup{R}}_{\mathbf{\Sigma}_{i}}(\alpha_{i}\mathbf{D}_{i})}\mathcal{M}$, the second Wolfe condition \eqref{41} can be converted to 
\begin{equation}
    |\langle\textup{grad}_{\mathcal{M}}L(\textbf{\textup{R}}_{\mathbf{\Sigma}_{i}}(\alpha_{i}\mathbf{D}_{i})),\mathbf{D}_{i}\rangle| \leq -c_{2}\langle \textup{grad}_{\mathcal{M}}L(\mathbf{\Sigma}_{i}),\mathbf{D}_{i}\rangle. \label{42}
\end{equation}
It has been shown that the step size satisfying the strong Wolfe conditions always exists if $\text{grad}_{\mathcal{M}}L(\textbf{R}_{\mathbf{\Sigma}}(\mathbf{D}))$ is Lipschitz continuous along $\mathbf{D}$ \cite[Proposition 1]{ring2012optimization}. In the following proposition, we show the Lipschitz continuity of the Riemannian gradient.
\begin{Prop}
    The objective function $L(\mathbf{\Sigma})$ is bounded below, meaning that there exists a constant $L^{*}\in\mathbb{R}$ such that for all $\mathbf{\Sigma}\in\mathcal{M}$, $L^{*}\leq L(\mathbf{\Sigma})$. Also, the Riemannian gradient $\text{grad}_{\mathcal{M}}L(\textbf{\textup{R}}_{\mathbf{\Sigma}}(\mathbf{D}))$ is Lipschitz continuous along $\mathbf{D}$.
    That is, for every $\mathbf{\Sigma}\in\mathcal{M}$, there exists a $K>0$ such that for all $\mathbf{D}\in\mathcal{T}_{\mathbf{\Sigma}}\mathcal{M}$,
    \begin{equation}
        \| \textup{grad}_{\mathcal{M}}L(\textbf{\textup{R}}_{\mathbf{\Sigma}}(\mathbf{D})) - \textup{grad}_{\mathcal{M}}L(\textbf{\textup{R}}_{\mathbf{\Sigma}}(\mathbf{0})) \| \leq  K\|\mathbf{D}\|. \label{43}
    \end{equation}
\end{Prop}
\begin{proof}
See Appendix A.
\end{proof}
Since the Riemannian gradient does not converge if its derivative is unbounded, the Lipschitz continuity of Riemannian gradient is crucial for the convergence of RCG method. Using the strong Wolfe conditions and the Lipschitz continuity of Riemannian gradient, we can show that the angle between the Riemannian gradient and the conjugate direction is bounded. 
\begin{thm}[Zoutendijk condition]
    Let $\mathbf{\Sigma}_{i}\in\mathcal{M}$ be the point and $\mathbf{D}_{i}\in\mathcal{T}_{\mathbf{\Sigma}_{i}}\mathcal{M}$ be the conjugate direction of $i$-th RCG iteration. 
    Then
    \begin{equation}
        \sum_{i=1}^{\infty}\frac{\langle\textup{grad}_{\mathcal{M}}L(\mathbf{\Sigma}_{i}),\mathbf{D}_{i}\rangle^{2}}{\|\mathbf{D}_{i}\|^{2}} < \infty. \label{44}
    \end{equation}
\end{thm}
\begin{proof}
See Appendix B.
\end{proof}

Next, in the following proposition, we prove that the inner product of Riemannian gradient $\text{grad}_{\mathcal{M}}L(\mathbf{\Sigma})$ and the conjugate direction $\mathbf{D}_{i}$ is bounded.
\begin{Prop}
   Let $\mathbf{\Sigma}_{i}\in\mathcal{M}$ be the point and $\mathbf{D}_{i}\in\mathcal{T}_{\mathbf{\Sigma}_{i}}\mathcal{M}$ be the conjugate direction of $i$-th RCG iteration with the conjugate gradient parameter $\beta_{i} = \frac{\|\textup{grad}_{\mathcal{M}}L(\mathbf{\Sigma}_{i})\|^{2}}{\|\textup{grad}_{\mathcal{M}}L(\mathbf{\Sigma}_{i-1})\|^{2}}$. 
   If the step size $\alpha_{i}$ satisfies the strong Wolfe conditions \eqref{40} and \eqref{41} with $c_{2} < 1/2$, then for every $i\in\mathbb{N}$,
    \begin{equation}
        -\frac{1}{1 - c_{2}} \leq \frac{\langle\textup{grad}_{\mathcal{M}}L(\mathbf{\Sigma}_{i}),\mathbf{D}_{i}\rangle}{\|\textup{grad}_{\mathcal{M}}L(\mathbf{\Sigma}_{i})\|^{2}} \leq  \frac{2c_{2} - 1}{1 - c_{2}}. \label{45}
    \end{equation}
\end{Prop}
\begin{proof}
See Appendix C.
\end{proof}
Finally, by combining Theorem 1 and Proposition 2, we show the convergence of RCG method. 
\begin{thm}(Convergence of RCG method)
    Let $\mathbf{\Sigma}_{i}\in\mathcal{M}$ be the point and $\mathbf{D}_{i}\in\mathcal{T}_{\mathbf{\Sigma}_{i}}\mathcal{M}$ be the conjugate direction of $i$-th RCG iteration with the conjugate gradient parameter $\beta_{i}=\frac{\|\textup{grad}_{\mathcal{M}}L(\mathbf{\Sigma}_{i})\|^{2}}{\|\textup{grad}_{\mathcal{M}}L(\mathbf{\Sigma}_{i-1})\|^{2}}$. 
    If the step size $\alpha_{i}$ satisfies the strong Wolfe conditions \eqref{40} and \eqref{41} with $c_{2} < 1/2$, then 
    \begin{equation}
        \liminf_{i\to\infty}\|\textup{grad}_{\mathcal{M}}L(\mathbf{\Sigma}_{i})\|=0. \label{46}
    \end{equation}
\end{thm}
\begin{proof}
See Appendix D.
\end{proof}

\begin{remark}
Theorem 2 shows that the RCG iteration of proposed scheme converges to a local minimizer $\mathbf{\Sigma}^{*}$ of objective function $L(\mathbf{\Sigma})$ on the product manifold $\mathcal{M}$.
\end{remark}

\subsection{Computational Complexity Analysis of RCG-JO Algorithm}
In our analysis, we measure the complexity in terms of the number of floating point operations (flops). We first provide the complexity analysis of the joint optimization of $\boldsymbol{\theta}$ and $\mathbf{W}$ in Algorithm 2.
\begin{lemma}
The total computational complexity $\mathcal{C}_{(\boldsymbol{\theta}, \mathbf{W})}$ of \textup{Algorithm 2} is given by
\begin{equation}
\mathcal{C}_{(\boldsymbol{\theta}, \mathbf{W})} = \mathcal{O}(K^{2}N^{2}M + K^{2}N^{3} + K^{2}M^{2}). \label{54}
\end{equation}
\end{lemma}
\begin{proof}
See Appendix E.
\end{proof}

After updating the RIS phase shift vector $\boldsymbol{\theta}$ and the BS beamforming matrix $\mathbf{W}$, we update the Lagrangian multiplier $\boldsymbol{\lambda}$ using the subgradient method.
The complexity of updating $\boldsymbol{\lambda}$ is $\mathcal{C}_{\boldsymbol{\lambda}} = \mathcal{O}(K^{2}N^{2}M)$.

Once $\boldsymbol{\theta}$, $\mathbf{W}$, and $\boldsymbol{\lambda}$ are updated, the optimization of the uplink power vector $\mathbf{p}$ is achieved by solving an LP problem in \eqref{39}.
Note that the numbers of flops to compute $A_{j,k}(\boldsymbol{\theta},\mathbf{W})$ and to solve the LP problem are $N^{2}M$ and $K^{3}$, respectively.
Thus, the overall complexity to optimize $\mathbf{p}$ is $\mathcal{C}_{\mathbf{p}} = \mathcal{O}(K^{2}N^{2}M + K^{3})$.

In conclusion, the computational complexity $\mathcal{C}_{\text{RCG-JO}}$ of the proposed RCG-JO scheme is 
\begin{align}
\mathcal{C}_{\text{RCG-JO}} 
&= \mathcal{C}_{(\boldsymbol{\theta}, \mathbf{W})} + \mathcal{C}_{\boldsymbol{\lambda}} + \mathcal{C}_{\mathbf{p}} \label{55}\\
&= \mathcal{O}(K^{2}N^{2}M + K^{3} + K^{2}N^{3} +  K^{2}M^{2}). \label{56}
\end{align}

\begin{table}[t!]
\begin{center}
\caption{Comparison of computational complexity ($K = 1$, $M = 4$)}
\centering
 \begin{tabular}{|c|c|c|c|c|c|} 
 \hline
 \multirow{2}{*}{ } & 
 \multirow{2}{*}{\shortstack[l]{Number of floating point\\  operations (flops) per iteration}} &
 \multicolumn{4}{c|}{Complexity for various numbers of reflecting elements $N$} \\
 \cline{3-6} 
 &  & $N = 4$ & $N = 16$ & $N = 32$ & $N = 64$ \\
 \hline
 \textbf{RCG-JO} & $\mathcal{O}(K^{2}N^{2}M + K^{3} + K^{2}N^{3} +  K^{2}M^{2})$ & $1.45 \times 10^{2}$ & $5.14 \times 10^{3}$ & $3.69 \times 10^{4}$ & $2.79 \times 10^{5}$ \\
 \hline
 \textbf{SDR-based scheme} & $\mathcal{O}(K^{2}N^{2}M + N^{6} + K^{6} M^{6})$ & $8.26 \times 10^{3}$ & $1.68 \times 10^{7}$ & $1.07 \times 10^{9}$ & $6.87 \times 10^{10}$ \\
 \hline
\end{tabular}
\end{center}
\end{table}

For comparison, we also discuss the complexity of the conventional SDR-based scheme, which consists of three major steps: 1) optimization of $\boldsymbol{\theta}$ using SDR for fixed $\mathbf{w}$ and $\mathbf{p}$, 2) optimization of $\mathbf{w}$ using SDR for fixed $\boldsymbol{\theta}$ and $\mathbf{p}$, and 3) optimization of $\mathbf{p}$ by solving an LP problem with the obtained $\boldsymbol{\theta}$ and $\mathbf{W}$.
It has been shown that the worst-case complexity of the SDR method for optimizing $\boldsymbol{\theta} \in \mathbb{C}^{N \times 1}$ is $\mathcal{O}(N^{6})$ \cite{WuIntelligent2019}.
Similarly, the worst-case complexity of optimizing $\mathbf{W} \in \mathbb{C}^{M \times K}$ using SDR is $\mathcal{O}(K^{6} M^{6})$.
In addition, the optimization of $\mathbf{p}$ is achieved by solving the LP problem, so that the complexity is $\mathcal{O}(K^{2}N^{2}M + K^{3})$.
In summary, the overall complexity $\mathcal{C}_{\text{SDR}}$ of the SDR-based scheme is 
\begin{equation}
\mathcal{C}_{\text{SDR}} = \mathcal{O}(K^{2}N^{2}M + N^{6} + K^{6} M^{6}). \label{57}
\end{equation}

In Table I, we compare the computational complexities of the proposed RCG-JO and the conventional SDR-based schemes.
When the number of reflecting elements $N$ is $4$, the complexity of RCG-JO is significantly lower than that of the SDR-based scheme (around $98\%$).
This is because the essential operations of RCG-JO are just matrix additions and multiplications whereas the SDR-based scheme needs a convex optimization technique to solve the SDP problem.

\begin{figure}[!t]
\centering
\includegraphics[width=0.35\columnwidth]{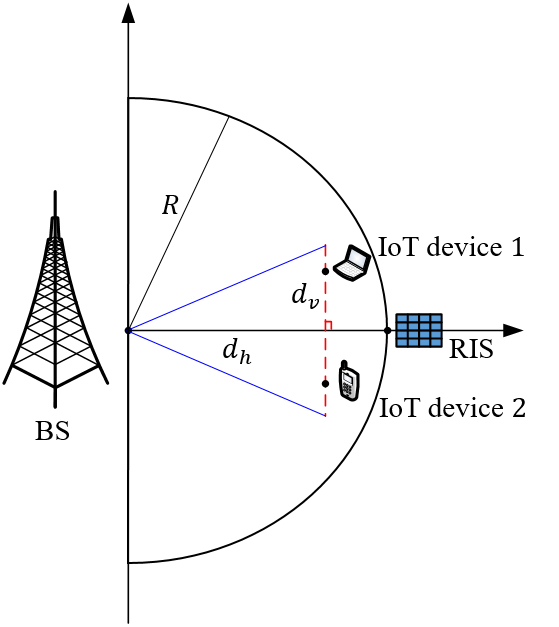}
\caption{Locations of BS, RIS, and IoT devices for $K = 2$.} \label{fig4}
\end{figure}

\section{Simulation Results}\label{VI}
\subsection{Simulation Setup}
In this section, we present the numerical results to evaluate the performance of proposed RCG-JO algorithm.
Our simulation setup is based on the uplink IoT network where $K = 2$ single-antenna devices transmit signals to the BS equipped with $M = 4$ receiving antennas (see Fig. 4).
This uplink transmission is assisted by the RIS equipped with $N = 64$ reflecting elements, which is randomly located at the circumference of the circle centered at the BS with the radius $R = 65\,$m.
Also, $K$ devices are uniformly distributed on a line that is perpendicular to the line connecting the BS and the RIS.
The maximal vertical distance between the IoT devices and the line is $d_{v} = 3\,$m.
We set the horizontal distance between the BS and IoT devices to $d_{h} = 57\,$m.
Throughout the simulations, we set the rate requirement $R_{k}^{\textup{min}}$ and noise power $\sigma_{k}^{2}$ to $0.3\,$bps/Hz and $-60\,$dBm, respectively. 
We use the Rician fading channel models for $\mathbf{d}_{k}$, $\mathbf{u}_{k}$ and $\mathbf{G}$ with the Rician factors $\kappa_{\mathbf{d}} = 0$, $\kappa_{k} = 10$, and $\kappa_{\mathbf{G}} = \infty$, respectively.
The path loss is $\beta = C_{0}(d/D_{0})^{-\alpha}$, where $d$ is the distance, $\alpha$ is the path loss exponent, and $C_{0} = -30\,\text{dB}$ is the path loss at the reference distance $D_{0} = 1\,\text{m}$ \cite{WuIntelligent2019}.
For the channels $\mathbf{d}_{k}$, $\mathbf{u}_{k}$ and $\mathbf{G}$, we set $\alpha$ to be $3.8$, $2.8$, and $2$, respectively. 
In each point of the plots, the simulation results are averaged over $1,000$ independent channel realizations.
The simulation parameters are summarized in Table II.

For comparison, we test the following benchmark techniques: 
1) SDR-based scheme where $\boldsymbol{\theta}$ and $\mathbf{W}$ are optimized alternately using SDR\footnote{Note that in the SDR-based scheme, after solving the relaxed SDP problem, the Gaussian randomization technique is employed to find out the feasible rank-one solution.} \cite{WuIntelligent2019},
2) deep reinforcement learning (DRL)-based scheme where $\boldsymbol{\theta}$ and $\mathbf{W}$ are jointly optimized by leveraging the DRL technique \cite{huang2020reconfigurable},
3) random phase shifts where $\boldsymbol{\theta}$ is randomly generated and the maximum ratio transmission (MRT) is used for $\mathbf{W}$, and
4) conventional system without RIS where $\mathbf{W}$ is generated randomly.

\begin{table}[t!]
\begin{center}
\caption{System parameters}
\begin{tabular}{|c|c|c|c|}
\hline
\textbf{Parameters} & \textbf{Values} & \textbf{Parameters} & \textbf{Values} \\ \hline
Number of devices ($K$) & $2$ & BS-devices vertical distance ($d_{v}$) & $3\,$m\\ \hline
Number of BS antennas ($M$) & $4$ & Carrier frequency ($f$) & $2.5\,$GHz  \\ \hline
Number of RIS reflecting elements ($N$) & $64$ & Noise power ($\sigma_{k}^{2}$) & $-60\,$dBm\\ \hline
BS-RIS distance ($R$) & $65\,$m & Rate requirement of devices  ($R_{k}^{\textup{min}}$) & $0.3\,$bps/Hz\\ \hline
BS-devices horizontal distance ($d_{h}$) & $57\,$m & Maximum transmission power ($p_{k}^{\text{max}}$) & $1\,$W\\ \hline
\end{tabular}
\end{center}
\end{table}

\subsection{Simulation Results}
In Fig. 5, we plot the uplink transmit power as a function of the number of RIS reflecting elements $N$.
From the simulation results, we observe that RCG-JO outperforms the conventional schemes using SDR and random phase shifts.
For example, when the number of RIS reflecting elements is $N = 64$, the proposed scheme achieves $44\%$ and $74\%$ reduction in power over the SDR-based scheme and the conventional scheme using random phase shifts, respectively.
Also, we see that the performance gap between RCG-JO and the SDR-based scheme increases gradually with $N$.
Furthermore, we see that the power saving gain of the conventional scheme using random phase shifts does not change with $N$.
This is because without the optimization of RIS phase shifts and BS beamforming vectors, the RIS reflected signal power is comparable to the signal power transmitted from the direct link, so that the gain obtained from the joint active and passive beamforming is marginal.
Also, we compare the uplink transmit power of RCG-JO with the DRL-based scheme.
We observe that RCG-JO achieves $57\%$ of power reduction over the DRL-based scheme.
Note, for the DRL-based approach, it is not easy to find out the optimal decision (i.e., RIS phase shifts and BS beamforming vectors) minimizing the uplink transmit power and at the same time satisfying the rate requirements of the IoT devices.
This is because the goal of DRL is to learn the decision policy maximizing the cumulative reward so that the minimization of uplink transmit power and the rate requirements might not be satisfied simultaneously.

\begin{figure}[!t]
\begin{minipage}{8cm}
\centering
\includegraphics[width=8.5cm]{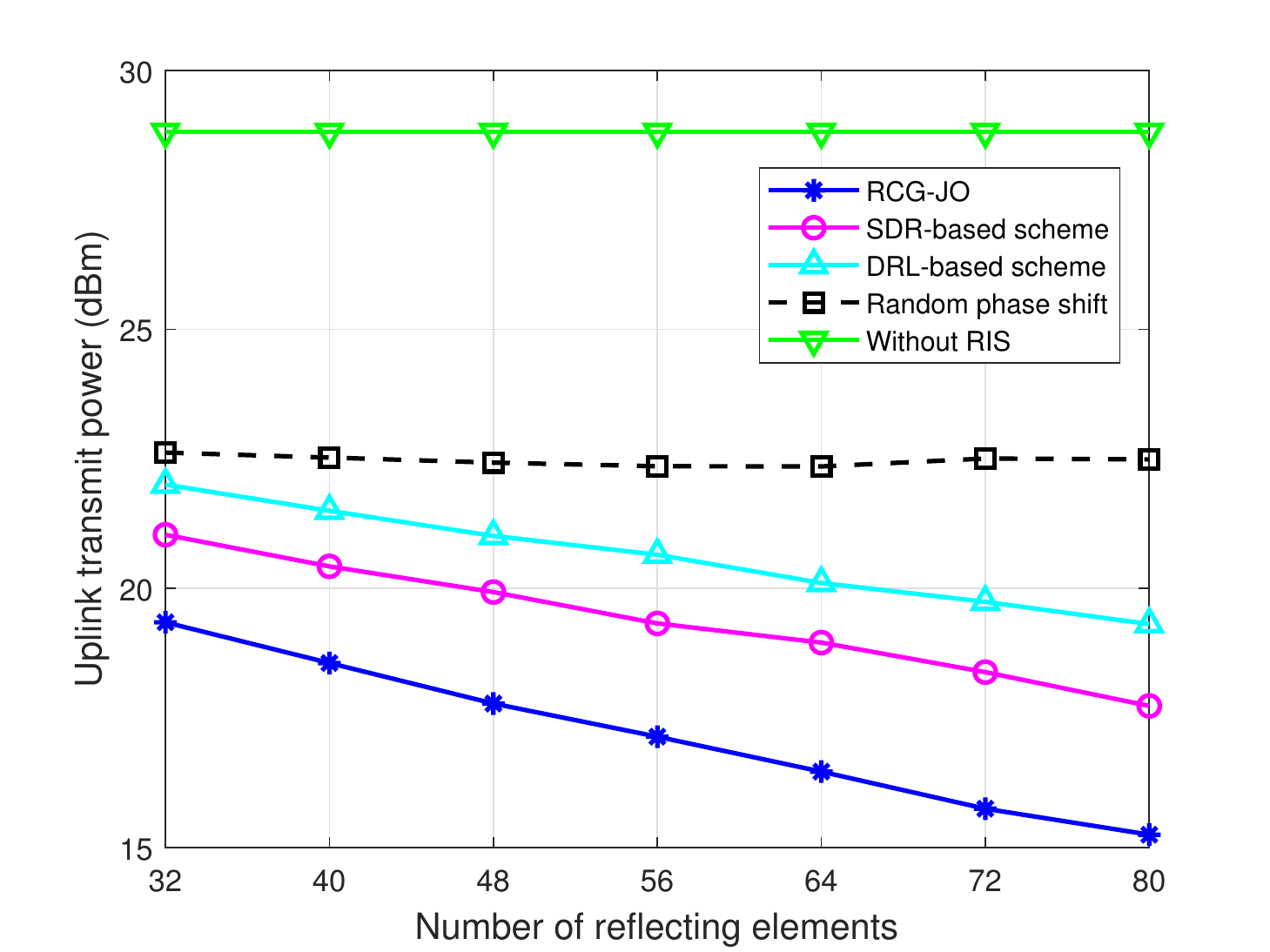}
\caption{Uplink transmit power vs. number of reflecting elements $N$.}
\end{minipage}
\begin{minipage}{8cm}
\centering
\includegraphics[width=8.5cm]{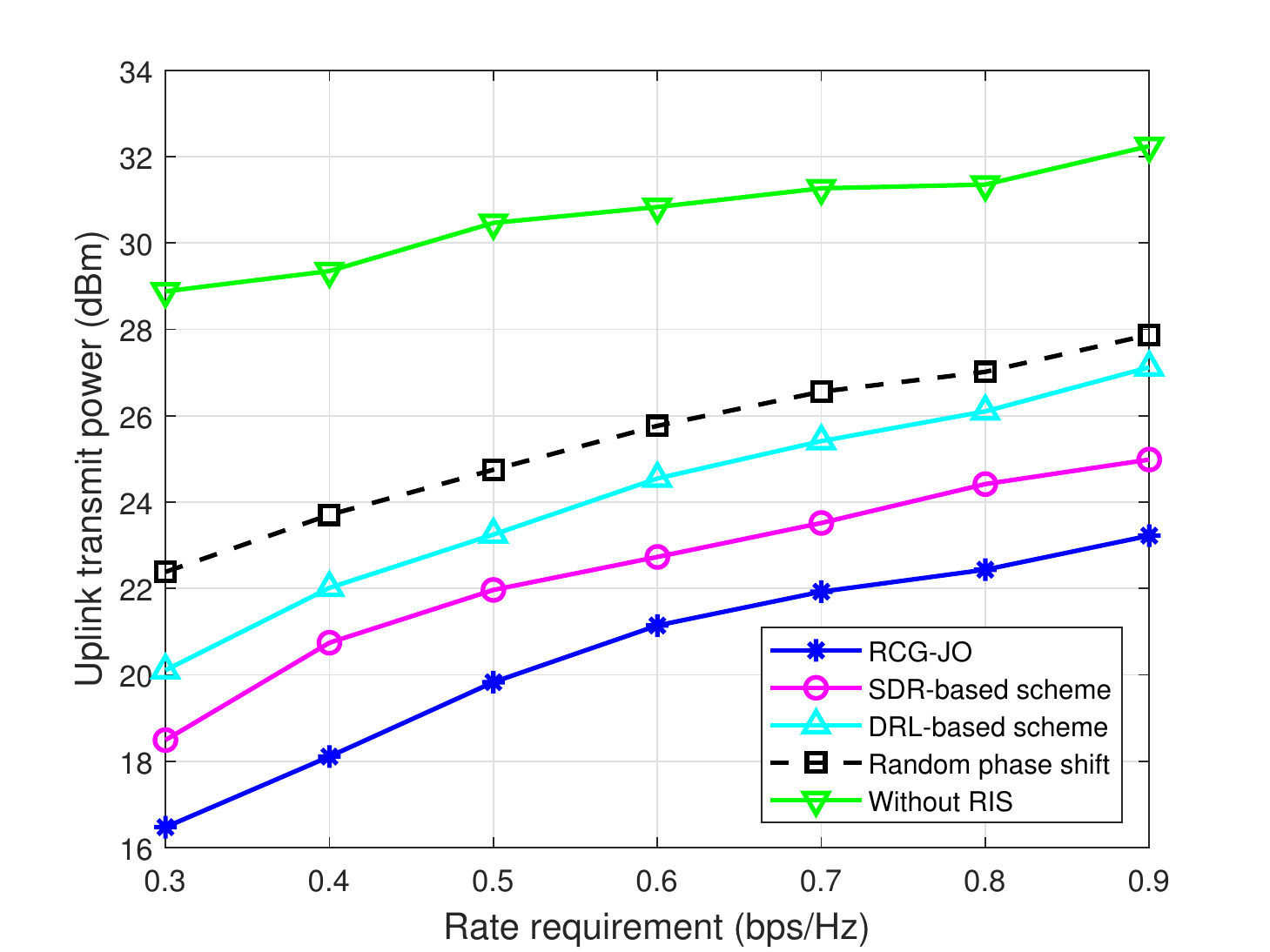}
\caption{Uplink transmit power vs. rate requirement of devices $R_{k}^{\textup{min}}$.}
\end{minipage}
\end{figure}

We next evaluate the uplink transmit power of the proposed RCG-JO algorithm and benchmark schemes as a function of the rate requirement of IoT devices $R_{k}^{\textup{min}}$.
As shown in Fig. 6, we observe that the uplink transmit power increases when the rate requirement of IoT devices becomes more strict.
Also, when compared to the case without RIS, the rate requirement of IoT devices can be satisfied with lower uplink transmit power in the proposed scheme.
For instance, when the rate requirement of devices is $R_{k}^{\textup{min}} = 0.4\,$bps/Hz, RCG-JO achieves around $92\%$ of uplink power reduction over the conventional scheme without RIS.
This is because through the joint active and passive beamforming at the BS and RIS, we can improve the signal power and reduce the interference so that the rate requirements of IoT devices can be satisfied with lower transmit power.
Furthermore, we see that RCG-JO significantly reduces the uplink transmit power over the benchmark schemes using SDR, DRL, and random phase shift.
When $R_{k}^{\textup{min}} = 0.8\,$bps/Hz, for example, RCG-JO saves almost $65\%$ of uplink transmit power over the conventional scheme using random phase shift.
Even when compared to the SDR-based and DRL-based schemes, the uplink power saving gains of RCG-JO are more than $37\%$ and $57\%$, respectively.

\begin{figure}[!t]
\begin{minipage}{8cm}
\centering
\includegraphics[width=8.5cm]{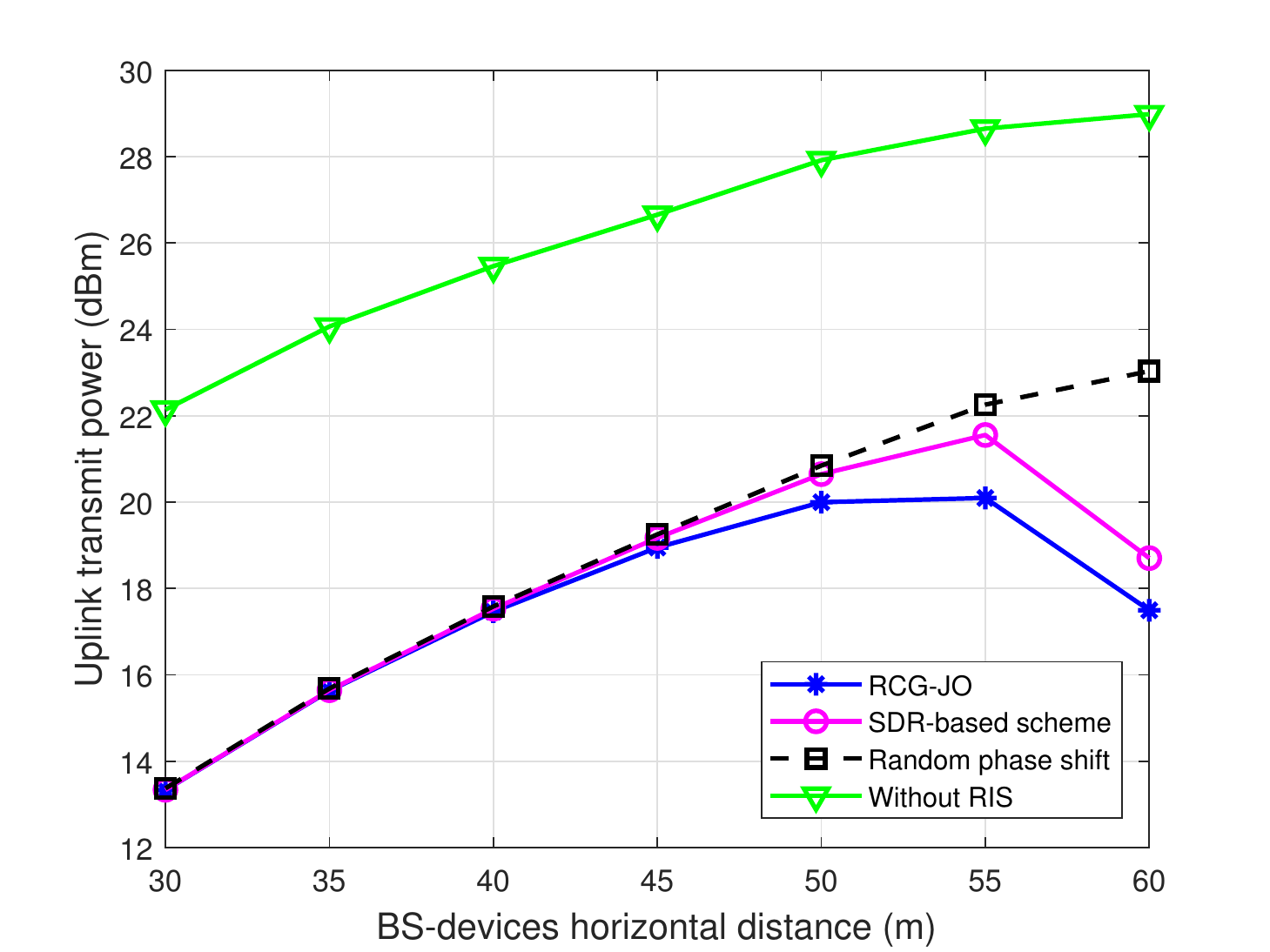}
\caption{Uplink transmit power vs. BS-devices horizontal distance $d_{h}$.}
\end{minipage}
\begin{minipage}{8cm}
\centering
\includegraphics[width=8.5cm]{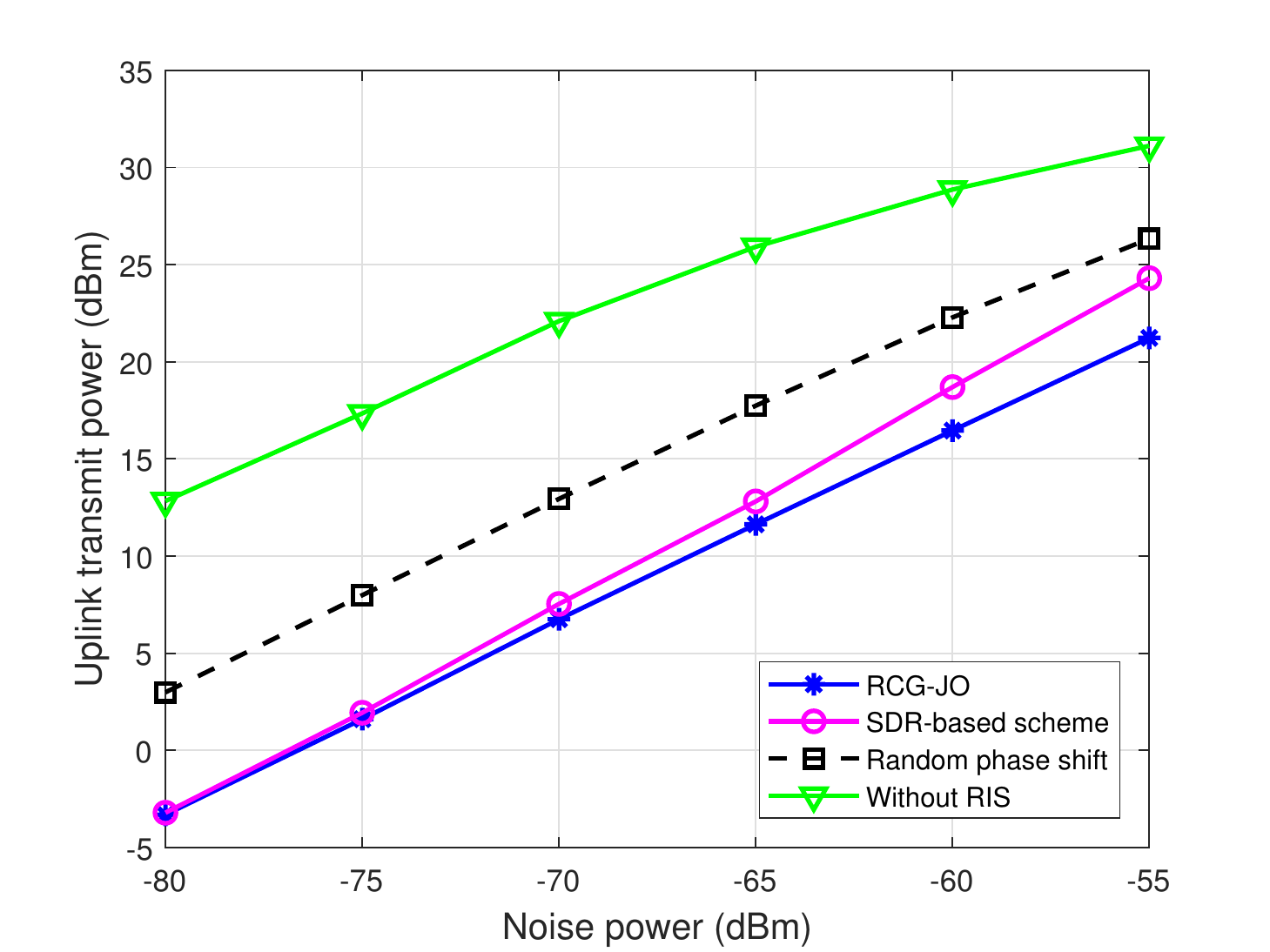}
\caption{Uplink transmit power vs. noise power $\sigma_{k}^{2}$.} 
\end{minipage}
\end{figure}

In Fig. 7, we plot the uplink transmit power as a function of the horizontal distance between the BS and IoT devices $d_{h}$.
We observe that RCG-JO outperforms the benchmark schemes as $d_{h}$ increases.
For example, when the horizontal distance between the BS and devices is $d_{h} = 55\,$m, RCG-JO saves $29$\% and $40$\% of the uplink transmit power over the SDR-based scheme and the conventional scheme using random phase shift, respectively.
In the system where the RIS is not employed, we see that the uplink transmit power is considerable due to the signal attenuation, in particular when the devices are located far away from the BS.
Whereas, in the RIS-aided IoT networks, we see that the uplink transmit power increases initially and then decreases as the horizontal distance increases.
Basically, when the IoT device is far from the BS, large uplink transmit power is needed to satisfy the rate requirement of a device.
In this case, the IoT device is getting close to the RIS so that the RIS-aided channel gain increases gradually.
As a result, the rate requirement can be satisfied with relatively lower uplink transmit power.

\begin{figure}[!t]
\begin{minipage}{8cm}
\centering
\includegraphics[width=8.5cm]{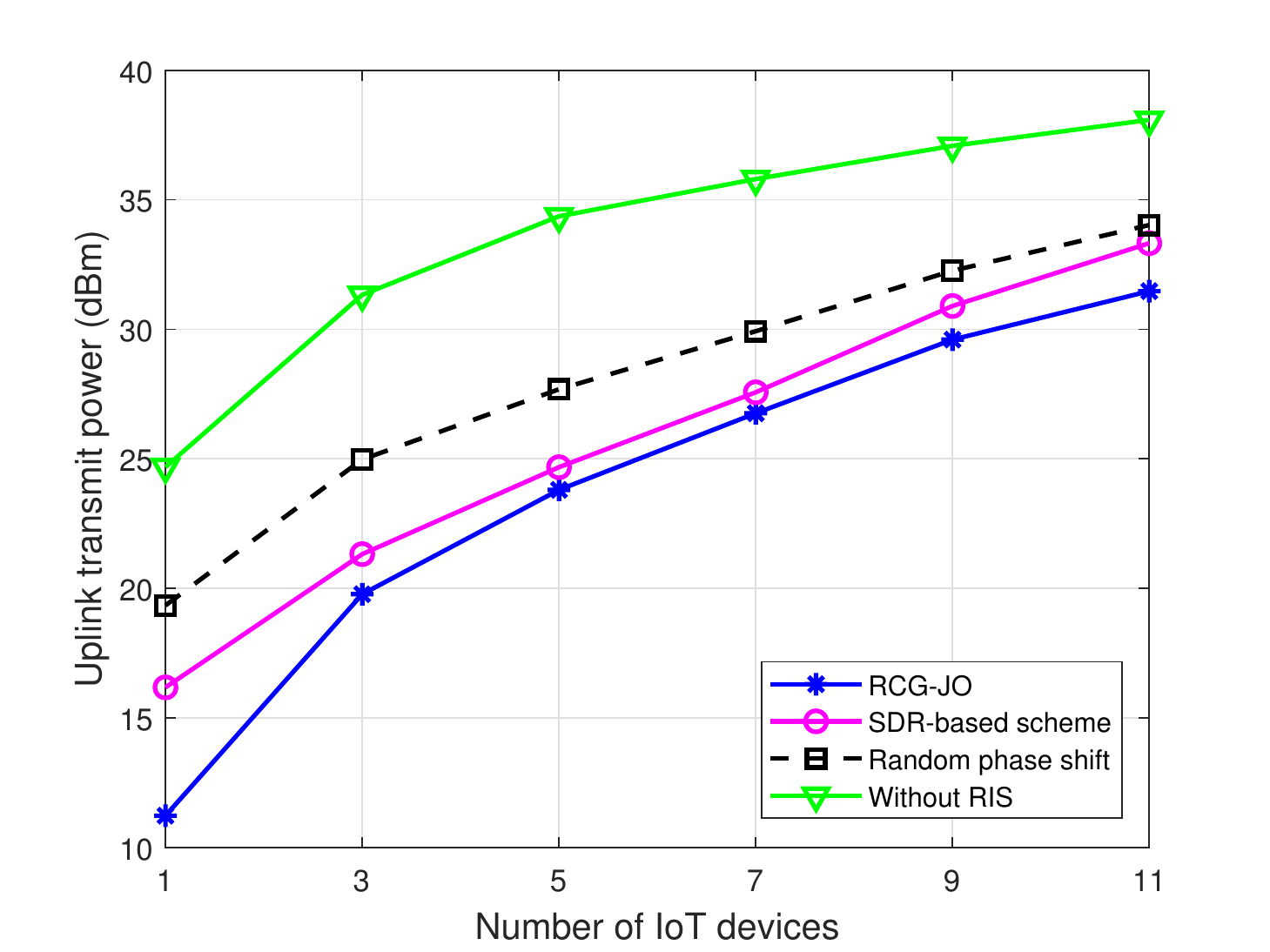}
\caption{Uplink transmit power vs. number of IoT devices $K$.} 
\end{minipage}
\begin{minipage}{8cm}
\centering
\includegraphics[width=8.5cm]{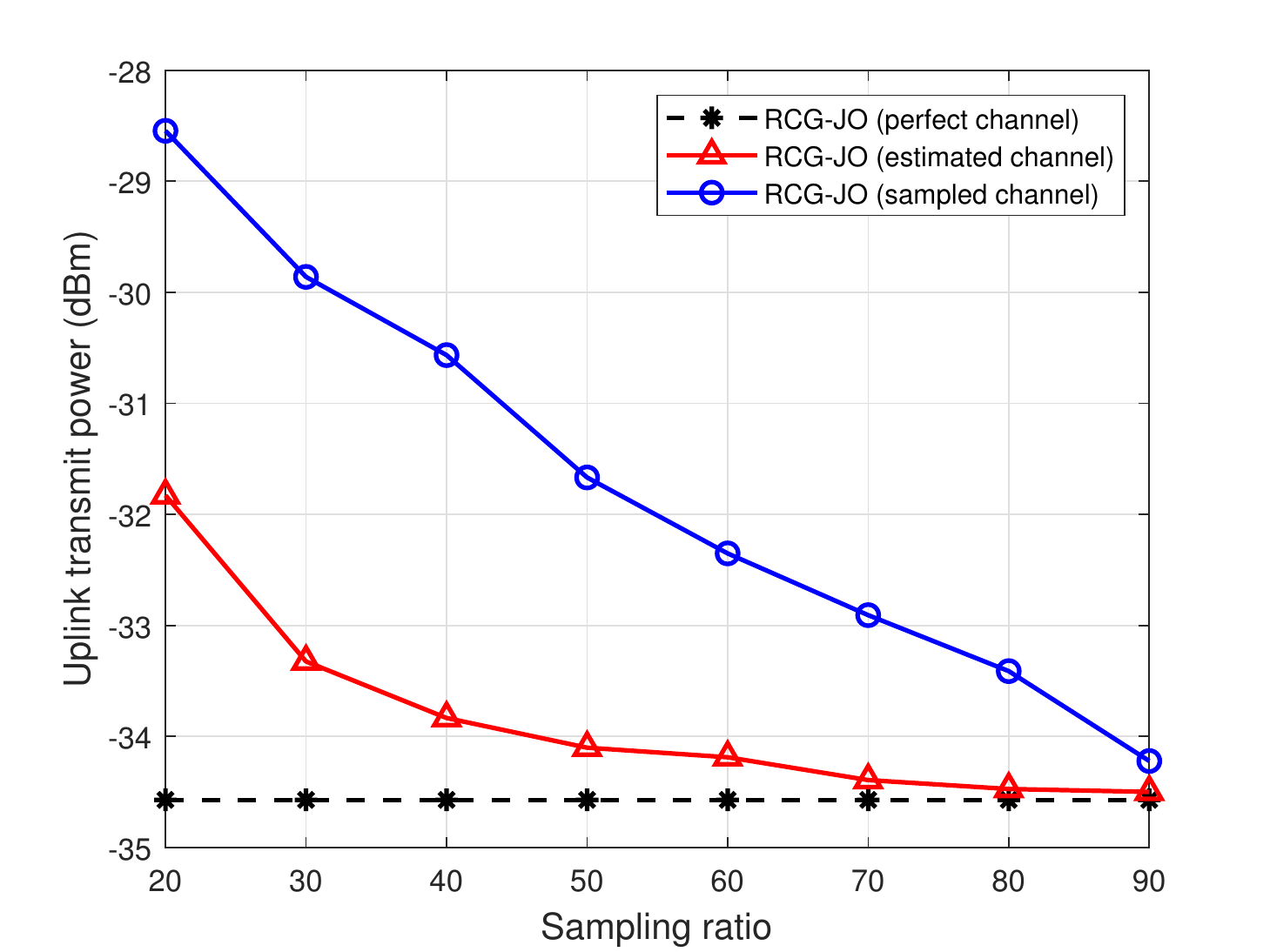}
\caption{Uplink transmit power vs. percentage of active reflecting elements.} 
\end{minipage}
\end{figure}

In Fig. 8, we investigate the uplink transmit power of the proposed scheme and benchmark schemes as a function of the noise power $\sigma_{k}^{2}$.
From the simulation results, we observe that the uplink transmit power increases with the noise power $\sigma_{k}^{2}$ and RCG-JO outperforms the benchmark schemes.
For example, when the noise power $\sigma_{k}^{2} = -55\,$dBm, the proposed scheme achieves $51\%$ and $70\%$ reduction over the conventional schemes using SDR and the random phase shift, respectively.
We also see that by using the RIS, the uplink transmit power of RCG-JO is significantly reduced (more than $80\%$) over the conventional scheme without RIS.

In Fig. 9, we evaluate the uplink transmit power of RCG-JO and benchmark schemes as the number of IoT devices $K$ increases. 
We observe that RCG-JO outperforms the benchmark schemes in all tested scenarios.
For example, when $K = 11$, RCG-JO saves $35\%$ and $78\%$ of the uplink transmit power over the SDR-based scheme and the conventional scheme without RIS, respectively. 
Also, we see that the power saving gain of RCG-JO over the SDR-based scheme increases with $K$. 
For instance, when we change $K$ from $5$ to $11$, the power reduction gain of RCG-JO over the SDR-based scheme increases from $18\%$ to $35\%$.
While RCG-JO solves the unconstrained optimization problem on the product manifold, the SDR-based scheme needs to find out the feasible rank-one solution after solving the SDP problem so that the performance degradation is inevitable.
This, together with the result of computational complexity analysis in Section IV.B, demonstrates the effectiveness of RCG-JO.

To evaluate the effectiveness of proposed channel estimation technique, we investigate the uplink transmit power of RCG-JO using the perfect channel information, the estimated channel information, and the sampled channel information when $N = 100$.
In Fig. 10, we observe that when the percentage of active reflecting elements is larger than $20\%$, RCG-JO using the estimated channel information performs close to RCG-JO using the genie channel information. 
This is because the RIS-aided channel matrix can be readily modeled as a low-rank matrix and thus the LRMC algorithm can effectively reconstruct the full channel matrix (see Section II.B).
In Fig. 11, we investigate the cumulative distributions of the number of iterations required for the convergence of outer iteration (Algorithm 1) and inner iteration (Algorithm 2).
We test the number of iterations required for the convergence when using $10,000$ independent channel realizations.
In all tested cases, we observe that Algorithm 1 converges within $15$ iterations and Algorithm 2 converges within $30$ iterations.

\begin{figure}[!t]
\centering
\includegraphics[width=8.5cm]{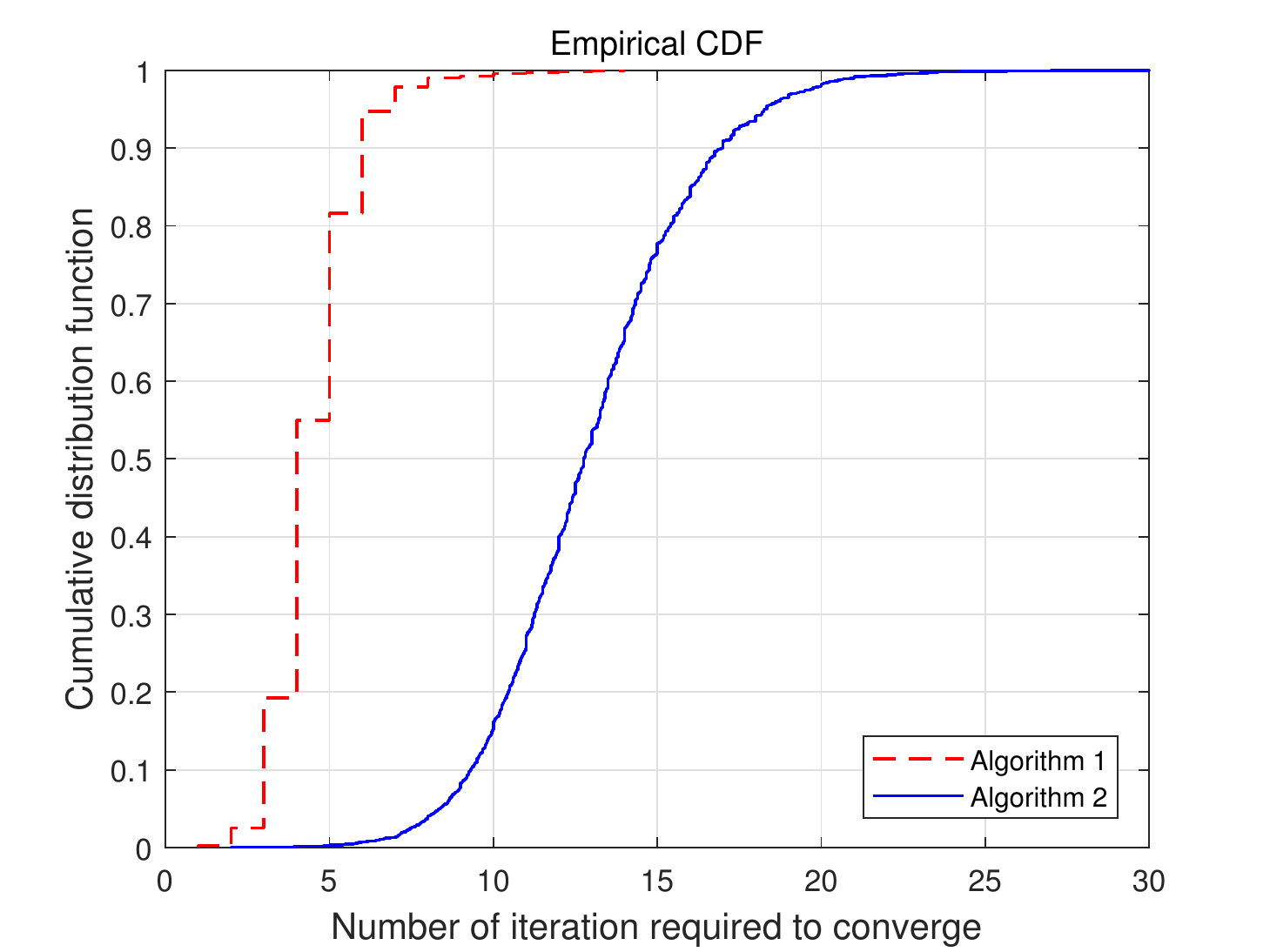}
\caption{Cumulative distribution of the number of iterations required to converge. \\} 
\end{figure}

\section{Conclusion}
In this paper, we proposed a RIS phase shift and BS beamforming optimization technique to minimize the uplink transmit power of a RIS-aided IoT network.
Key idea of the proposed RCG-JO algorithm is to jointly optimize the RIS phase shifts and BS beamforming vectors using the Riemannian conjugate gradient method.
By exploiting the product Riemannian manifold structure of the sets of unit-modulus RIS phase shift and unit-norm BS beamforming vector, we converted the uplink power minimization problem to an unconstrained problem on the Riemannian manifold.
Then, we employed the Riemannian conjugate gradient method to find out the optimal RIS phase shifts and the BS beamforming vectors simultaneously.
We demonstrated from the performance analysis and numerical evaluations that the proposed RCG-JO algorithm is effective in saving the uplink transmit power of RIS-aided IoT networks.
In our work, we assumed the ideal phase shift model where the reflection amplitude and phase shift are independent, but an extension to the realistic scenarios with the phase-dependent reflection amplitude would also be an interesting research direction worth pursuing.

\section*{Appendix A\\Proof of Proposition 1}
One can easily see that due to the unit-modulus constraint of $\boldsymbol{\theta}$ \eqref{12c} and the unit-norm constraint of $\mathbf{W}$ \eqref{12d}, the objective function $L(\mathbf{\Sigma})$ is bounded.

Recall that the Riemannian gradient can be decomposed as $\text{grad}_{\mathcal{M}}L(\textbf{R}_{\boldsymbol{\Sigma}}(\bar{\mathbf{D}}) = \text{grad}_{\mathcal{M}_{\boldsymbol{\theta}}}L( \textbf{R}_{\boldsymbol{\theta}}(\mathbf{d})) \oplus \text{grad}_{\mathcal{M}_{\mathbf{W}}}L(\textbf{R}_{\mathbf{W}}(\mathbf{D}))$ where $\bar{\mathbf{D}} = \mathbf{d} \oplus \mathbf{D}$. Thus, we firstly find out the constants $K_{\boldsymbol{\theta}}$ and $K_{\mathbf{W}}$ satisfying the Lipschitz conditions for $\text{grad}_{\mathcal{M}_{\boldsymbol{\theta}}}L( \textbf{R}_{\boldsymbol{\theta}}(\mathbf{d}))$ and $\text{grad}_{\mathcal{M}_{\mathbf{W}}}L(\textbf{R}_{\mathbf{W}}(\mathbf{D}))$, respectively.
We then obtain the constant $K = \max(K_{\boldsymbol{\theta}},K_{\mathbf{W}})$ satisfying the Lipschitz condition for $\text{grad}_{\mathcal{M}}L(\textbf{R}_{\boldsymbol{\Sigma}}(\bar{\mathbf{D}}))$. 

When $\mathbf{W}$ is given, $L(\boldsymbol{\theta})$ is a quadratic function with respect to $\boldsymbol{\theta}$:
\begin{align}
    L(\boldsymbol{\theta}) 
    &= \sum_{k=1}^{K} \lambda_{k} \Big( - \frac{p_{k}}{2^{R_{k}^{\textup{min}}} -1} A_{k,k}(\boldsymbol{\theta},\mathbf{w}_{k}) + \sum_{j \neq k}^{K} p_{j} A_{j,k}(\boldsymbol{\theta},\mathbf{w}_{k}) + \sigma_{k}^{2} \Big) \label{58}\\
    &= \boldsymbol{\theta}^{\text{H}}\mathbf{B}\boldsymbol{\theta} + \mathbf{b}^{\text{H}}\boldsymbol{\theta} +\boldsymbol{\theta}^{\text{H}}\mathbf{b} + b, \label{59}
\end{align}
where $\mathbf{B} = \sum_{k=1}^{K}\lambda_{k}\big(-\frac{p_{k}}{2^{R_{k}^{\text{min}}}-1}\mathbf{H}_{k}^{\text{H}}\mathbf{G}^{\text{H}}\mathbf{w}_{k}\mathbf{w}_{k}^{\text{H}}\mathbf{G}\mathbf{H}_{k} + \sum_{j\neq k}^{K}p_{j}\mathbf{H}_{j}^{\text{H}}\mathbf{G}^{\text{H}}\mathbf{w}_{k}\mathbf{w}_{k}^{\text{H}}\mathbf{G}\mathbf{H}_{j}\big)$ and $\mathbf{b} = \sum_{k=1}^{K}\lambda_{k}\big(-\frac{p_{k}}{2^{R_{k}^{\text{min}}}-1}\mathbf{H}_{k}^{\text{H}}\mathbf{G}^{\text{H}}\mathbf{w}_{k}\mathbf{w}_{k}^{\text{H}}\mathbf{d}_{k} + \sum_{j\neq k}^{K}p_{j}\mathbf{H}_{j}^{\text{H}}\mathbf{G}^{\text{H}}\mathbf{w}_{k}\mathbf{w}_{k}^{\text{H}}\mathbf{d}_{j}\big)$. Then $\text{grad}_{\mathcal{M}_{\boldsymbol{\theta}}}L(\textbf{R}_{\boldsymbol{\theta}}(\mathbf{d}))$ is expressed as
\begin{align}
    \text{grad}_{\mathcal{M}_{\boldsymbol{\theta}}}L(\textbf{R}_{\boldsymbol{\theta}}(\mathbf{d})) 
    &= \textbf{P}_{\mathcal{T}_{\boldsymbol{\theta}}\mathcal{M}_{\boldsymbol{\theta}}}(\nabla_{\boldsymbol{\theta}}L(\textbf{R}_{\boldsymbol{\theta}}(\mathbf{d}))) \label{60}\\
    &= \textbf{P}_{\mathcal{T}_{\boldsymbol{\theta}}\mathcal{M}_{\boldsymbol{\theta}}}(\mathbf{B}\textbf{R}_{\boldsymbol{\theta}}(\mathbf{d}) + \mathbf{b}) \label{61}\\
    &= \mathbf{B}\textbf{R}_{\boldsymbol{\theta}}(\mathbf{d})+\mathbf{b} - \Re\{\boldsymbol{\theta}\odot(\mathbf{B}\textbf{R}_{\boldsymbol{\theta}}(\mathbf{d})+\mathbf{b})\}\odot\boldsymbol{\theta}. \label{62}
\end{align}
Using the triangle inequality, we have
\begin{align}
    \|\textup{grad}_{\mathcal{M}_{\boldsymbol{\theta}}}L(\textbf{R}_{\boldsymbol{\theta}}(\mathbf{d})) - \textup{grad}_{\mathcal{M}_{\boldsymbol{\theta}}}L(\textbf{R}_{\boldsymbol{\theta}}(\mathbf{0}))\|
    \leq& \|\mathbf{B}(\textbf{R}_{\boldsymbol{\theta}}(\mathbf{d})-\textbf{R}_{\boldsymbol{\theta}}(\mathbf{0}))\|\nonumber\\
    &+ \|\Re\{\boldsymbol{\theta}\odot(\mathbf{B}(\textbf{R}_{\boldsymbol{\theta}}(\mathbf{d})-\textbf{R}_{\boldsymbol{\theta}}(\mathbf{0})))\}\odot\boldsymbol{\theta}\| \label{63}\\
    \leq& 2\|\mathbf{B}(\textbf{R}_{\boldsymbol{\theta}}(\mathbf{d})-\textbf{R}_{\boldsymbol{\theta}}(\mathbf{0}))\| \label{64}\\
    \leq& 2\|\mathbf{B}\| \|\textbf{R}_{\boldsymbol{\theta}}(\mathbf{d})-\textbf{R}_{\boldsymbol{\theta}}(\mathbf{0})\|. \label{65}
\end{align}
Since $\mathbf{B}$ is a quadratic function of $\mathbf{W}$ and the elements of $\mathbf{W}$ are bounded by the unit-norm constraints, $\|\mathbf{B}\|$ is bounded on $\mathbf{W}\in\mathcal{M}_{\mathbf{W}}$.

Now, what we need to show is that the retraction operator $\textbf{R}_{\boldsymbol{\theta}}$ is Lipschitz continuous. 
To do so, we prove that each element of $\textbf{R}_{\boldsymbol{\theta}}$ is Lipschitz continuous. Let $\theta_{n}$ and $d_{n}$ be the $n$-th elements of $\boldsymbol{\theta}$ and $\mathbf{d}$, respectively. 
Since $|\theta_{n}| = 1$, we have
\begin{equation}
    |R_{\theta_{n}}(d_{n})-R_{\theta_{n}}(0)| = \bigg|\frac{\theta_{n} + d_{n}}{|\theta_{n} + d_{n}|} - \theta_{n}\bigg| =  |\theta_{n}|\bigg|\frac{1 + d_{n}\theta_{n}^{-1}}{|\theta_{n}||1 + d_{n}\theta_{n}^{-1}|} - 1\bigg| = \bigg|\frac{1 + d_{n}\theta_{n}^{-1}}{|1 + d_{n}\theta_{n}^{-1}|} - 1\bigg|. \label{66}
\end{equation}
Without the loss of generality, we assume that $\theta_{n}=1$. If $|d_{n}|\geq \frac{1}{2}$, then we have 
\begin{equation}
    \bigg|\frac{1 + d_{n}}{|1 + d_{n}|} - 1\bigg| \leq \bigg|\frac{1 + d_{n}}{|1 + d_{n}|}\bigg| + 1 = 2 \leq 4|d_{n}|. \label{67}
\end{equation}
Also, if $|d_{n}|<\frac{1}{2}$, then we have
\begin{equation}
    \bigg|\frac{1 + d_{n}}{|1 + d_{n}|} - 1\bigg|^{2} 
    = \frac{|1 + d_{n}| - \Re\{1 + d_{n}\}}{|1 + d_{n}|}
    = \frac{2\Im^{2}\{1 + d_{n}\}}{|1 + d_{n}|(|1 + d_{n}| + \Re\{1 + d_{n}\})}
    \leq \frac{2|d_{n}|^{2}}{|1 + d_{n}|^{2}}
    \leq 8|d_{n}|^{2}. \label{68}
\end{equation}
Combining \eqref{65}, \eqref{67}, and \eqref{68}, we see that $K_{\boldsymbol{\theta}}=8\sup_{\mathbf{W}\in\mathcal{M}_{\mathbf{W}}}\|\mathbf{B}\|$ satisfies the Lipschitz condition for $\text{grad}_{\mathcal{M}_{\boldsymbol{\theta}}}L(\textbf{R}_{\boldsymbol{\theta}}(\mathbf{d}))$. 
Similarly, we can obtain $K_{\mathbf{W}}$ using the fact that the column vectors of $\mathbf{W}$ have unit-norm. 
Finally, we obtain the constant $K = \max(K_{\boldsymbol{\theta}},K_{\mathbf{W}})$ satisfying the Lipschitz condition for $\text{grad}_{\mathcal{M}}L(\textbf{R}_{\boldsymbol{\Sigma}}(\bar{\mathbf{D}}))$. 

\section*{Appendix B\\Proof of Theorem 1}
Using the second Wolfe condition \eqref{42} and the Lipschitz continuity \eqref{43}, we have
\begin{align}
    (c_{2} - 1)\langle \text{grad}_{\mathcal{M}}L(\mathbf{\Sigma}_{i}),\mathbf{D}_{i}\rangle 
    &\leq \langle \text{grad}_{\mathcal{M}}L(\textbf{R}_{\mathbf{\Sigma}_{i}}(\alpha_{i}\mathbf{D}_{i})) - \text{grad}_{\mathcal{M}}L(\mathbf{\Sigma}_{i}),\mathbf{D}_{i}\rangle \label{69}\\
    &= \langle \text{grad}_{\mathcal{M}}L(\textbf{R}_{\mathbf{\Sigma}_{i}}(\alpha_{i}\mathbf{D}_{i})) - \text{grad}_{\mathcal{M}}L(\textbf{R}_{\mathbf{\Sigma}_{i}}(\mathbf{0})),\mathbf{D}_{i}\rangle \label{70}\\
    &\leq \|\text{grad}_{\mathcal{M}}L(\textbf{R}_{\mathbf{\Sigma}_{i}}(\alpha_{i}\mathbf{D}_{i})) - \text{grad}_{\mathcal{M}}L(\textbf{R}_{\mathbf{\Sigma}_{i}}(\mathbf{0}))\|\|\mathbf{D}_{i}\| \label{71}\\
    &\leq \alpha_{i}K\|\mathbf{D}_{i}\|^{2}. \label{72}
\end{align}
Then from the first Wolfe condition \eqref{40} and \eqref{72}, we have
\begin{align}
    L(\mathbf{\Sigma}_{i+1}) = L(\textbf{R}_{\mathbf{\Sigma}_{i}}(\alpha_{i}\mathbf{D}_{i})) 
    &\leq L(\mathbf{\Sigma}_{i}) + c_{1}\alpha_{i}\langle\text{grad}_{\mathcal{M}}L(\mathbf{\Sigma}_{i})),\mathbf{D}_{i}\rangle \label{73}\\
    &\leq L(\mathbf{\Sigma}_{i}) - c_{1}\frac{1 - c_{2}}{K}\frac{\langle \text{grad}_{\mathcal{M}}L(\mathbf{\Sigma}_{i}),\mathbf{D}_{i}\rangle^{2}}{\|\mathbf{D}_{i}\|^{2}}. \label{74}
\end{align}
Finally, by combining \eqref{74} for $i=1,\cdots, I$, we have
\begin{equation}
    \sum_{i = 1}^{I}\frac{\langle \text{grad}_{\mathcal{M}}L(\mathbf{\Sigma}_{i}),\mathbf{D}_{i}\rangle^{2}}{\|\mathbf{D}_{i}\|^{2}}
    \leq \frac{K(L(\mathbf{\Sigma}_{1}) - L(\mathbf{\Sigma}_{I + 1}))}{c_{1}(1 - c_{2})}
    \stackrel{(a)}{\leq}\frac{K(L(\mathbf{\Sigma}_{1}) - L^{*})}{c_{1}(1 - c_{2})}, \label{75}
\end{equation}
where (a) is from the fact that $L(\mathbf{\Sigma})$ is bounded below by $L^{*}$. 
By taking the limit $I\to\infty$ of \eqref{75}, we obtain the desired result \eqref{44}. 

\section*{Appendix C\\Proof of Proposition 2}
We use the mathematical induction to prove Proposition 2. 
When $i=1$,  $\mathbf{D}_{1} = -\text{grad}_{\mathcal{M}}L(\mathbf{\Sigma}_{1})$ so that we can easily see that \eqref{45} holds. 
Now suppose that \eqref{45} holds for $i\geq 1$. 
By using the update equation of $\mathbf{D}_{i}$ in \eqref{28} and the Fletcher-Reeves conjugate gradient parameter $\beta_{i} = \frac{\|\text{grad}_{\mathcal{M}}L(\mathbf{\Sigma}_{i})\|^{2}}{\|\text{grad}_{\mathcal{M}}L(\mathbf{\Sigma}_{i + 1})\|^{2}}$, we have
\begin{align}
    \frac{\langle\text{grad}_{\mathcal{M}}L(\mathbf{\Sigma}_{i + 1}),\mathbf{D}_{i + 1}\rangle}{\|\text{grad}_{\mathcal{M}}L(\mathbf{\Sigma}_{i + 1})\|^{2}} \label{76}
    &= \frac{\langle\text{grad}_{\mathcal{M}}L(\mathbf{\Sigma}_{i + 1}), -\textup{grad}_{\mathcal{M}}L(\mathbf{\Sigma}_{i + 1}) + \beta_{i + 1} \textbf{\textup{P}}_{\mathcal{T}_{\mathbf{\Sigma}_{i + 1}}\mathcal{M}} (\mathbf{D}_{i})\rangle}{\|\text{grad}_{\mathcal{M}}L(\mathbf{\Sigma}_{i + 1})\|^{2}}\\
    &= -1 + \beta_{i + 1}\frac{\langle\text{grad}_{\mathcal{M}}L(\mathbf{\Sigma}_{i+1}),\textbf{\textup{P}}_{\mathcal{T}_{\mathbf{\Sigma}_{i+1}}\mathcal{M}} (\mathbf{D}_{i})\rangle}{\|\text{grad}_{\mathcal{M}}L(\mathbf{\Sigma}_{i+1})\|^{2}} \label{77}\\
    &= -1 + \frac{\langle\text{grad}_{\mathcal{M}}L(\textbf{R}_{\mathbf{\Sigma}_{i}}(\alpha_{i}\mathbf{D}_{i})),\mathbf{D}_{i}\rangle}{\|\text{grad}_{\mathcal{M}}L(\mathbf{\Sigma}_{i})\|^{2}}. \label{78}
\end{align}
Then, using the second Wolfe condition \eqref{42} and \eqref{78}, we obtain
\begin{equation}
    -1 + c_{2}\frac{\langle\text{grad}_{\mathcal{M}}L(\mathbf{\Sigma}_{i}),\mathbf{D}_{i}\rangle}{\|\text{grad}_{\mathcal{M}}L(\mathbf{\Sigma}_{i})\|^{2}}
    \leq \frac{\langle\text{grad}_{\mathcal{M}}L(\mathbf{\Sigma}_{i+1}),\mathbf{D}_{i+1}\rangle}{\|\text{grad}_{\mathcal{M}}L(\mathbf{\Sigma}_{i+1})\|^{2}}
    \leq -1 - c_{2}\frac{\langle\text{grad}_{\mathcal{M}}L(\mathbf{\Sigma}_{i}),\mathbf{D}_{i}\rangle}{\|\text{grad}_{\mathcal{M}}L(\mathbf{\Sigma}_{i})\|^{2}}. \label{79}
\end{equation}
By employing the induction hypothesis \eqref{45} for $i$, we see that \eqref{45} holds for $i+1$, which establishes the proposition 2.

\section*{Appendix D\\Proof of Theorem 2}
Using the update equation of $\mathbf{\Sigma}_{i}$ in \eqref{29}, \eqref{45}, and the second Wolfe condition \eqref{41}, we obtain
\begin{align}
    |\langle \text{grad}_{\mathcal{M}}L(\mathbf{\Sigma}_{i}), \textbf{P}_{\mathcal{T}_{\mathbf{\Sigma}_{i}}\mathcal{M}}(\mathbf{D}_{i - 1})\rangle|
    &=|\langle \text{grad}_{\mathcal{M}}L(\textbf{R}_{\mathbf{\Sigma}_{i}}(\alpha_{i - 1}\mathbf{D}_{i - 1})), \textbf{P}_{\mathcal{T}_{\textbf{R}_{\mathbf{\Sigma}_{i}}(\alpha_{i - 1}\mathbf{D}_{i - 1})}\mathcal{M}}(\mathbf{D}_{i - 1})\rangle| \label{80}\\
    &\leq -c_{2}\langle \text{grad}_{\mathcal{M}}L(\mathbf{\Sigma}_{i - 1}),\mathbf{D}_{i - 1}\rangle\label{81}\\
    &\leq \frac{c_{2}}{1 - c_{2}}\|\text{grad}_{\mathcal{M}}L(\mathbf{\Sigma}_{i - 1})\|^{2}. \label{82}
\end{align}
Then from the update equation of conjugate direction $\mathbf{D}_{i}$ in \eqref{28}, we have
\begin{align}
    \|\mathbf{D}_{i}\|^{2} 
    &= \|-\text{grad}_{\mathcal{M}}L(\mathbf{\Sigma}_{i}) + \beta_{i}\textbf{P}_{\mathcal{T}_{\mathbf{\Sigma}_{i}}\mathcal{M}}(\mathbf{D}_{i - 1})\|^{2} \label{83}\\
    &= \|\text{grad}_{\mathcal{M}}L(\mathbf{\Sigma}_{i})\|^{2} - 2\beta_{i}\langle\text{grad}_{\mathcal{M}}L(\mathbf{\Sigma}_{i}),\textbf{P}_{\mathcal{T}_{\mathbf{\Sigma}_{i}}\mathcal{M}}(\mathbf{D}_{i - 1})\rangle + \beta_{i}^{2}\|\textbf{P}_{\mathcal{T}_{\mathbf{\Sigma}_{i}}\mathcal{M}}(\mathbf{D}_{i - 1})\|^{2} \label{84}\\
    &\leq \|\text{grad}_{\mathcal{M}}L(\mathbf{\Sigma}_{i})\|^{2} + \frac{2c_{2}}{1 - c_{2}}\beta_{i}\|\text{grad}_{\mathcal{M}}L(\mathbf{\Sigma}_{i - 1})\|^{2} + \beta_{i}^{2}\|\mathbf{D}_{i - 1}\|^{2} \label{85}\\
    &= \|\text{grad}_{\mathcal{M}}L(\mathbf{\Sigma}_{i})\|^{2} + \frac{2c_{2}}{1 - c_{2}}\|\text{grad}_{\mathcal{M}}L(\mathbf{\Sigma}_{i})\|^{2} + \frac{\|\text{grad}_{\mathcal{M}}L(\mathbf{\Sigma}_{i})\|^{4}}{\|\text{grad}_{\mathcal{M}}L(\mathbf{\Sigma}_{i - 1})\|^{4}}\|\mathbf{D}_{i - 1}\|^{2} \label{86}\\
    &= \frac{1 + c_{2}}{1 - c_{2}}\|\text{grad}_{\mathcal{M}}L(\mathbf{\Sigma}_{i})\|^{2} + \frac{\|\text{grad}_{\mathcal{M}}L(\mathbf{\Sigma}_{i})\|^{4}}{\|\text{grad}_{\mathcal{M}}L(\mathbf{\Sigma}_{i - 1})\|^{4}}\|\mathbf{D}_{i - 1}\|^{2}. \label{87}
\end{align}
By sequentially applying \eqref{87} until $i = 1$, we obtain
\begin{equation}
    \|\mathbf{D}_{i}\|^{2} \leq \frac{1 + c_{2}}{1 - c_{2}}\|\text{grad}_{\mathcal{M}}L(\mathbf{\Sigma}_{i})\|^{4}\sum_{j=1}^{i}\frac{1}{\|\text{grad}_{\mathcal{M}}L(\mathbf{\Sigma}_{j})\|^{2}}. \label{88}
\end{equation}
If we assume $\liminf_{i\to\infty}\|\text{grad}_{\mathcal{M}}L(\mathbf{\Sigma}_{i})\|\neq 0$,  meaning that $\epsilon (> 0)$ such that $\|\text{grad}_{\mathcal{M}}L(\mathbf{\Sigma}_{i})\|\geq \epsilon$ for all $i\in\mathcal{N}$ exists, then we have
\begin{equation}
    \|\mathbf{D}_{i}\|^{2}\leq \frac{1 + c_{2}}{1 - c_{2}}\|\text{grad}_{\mathcal{M}}L(\mathbf{\Sigma}_{i})\|^{4}\frac{i}{\epsilon^{2}}. \label{89}
\end{equation} 
This implies that 
\begin{equation}
    \sum_{i=1}^{\infty}\frac{\|\text{grad}_{\mathcal{M}}L(\mathbf{\Sigma}_{i})\|^{4}}{\|\mathbf{D}_{i}\|^{2}}\geq \frac{1-c_{2}}{1+c_{2}}\epsilon^{2}\sum_{i=1}^{\infty}\frac{1}{i} = \infty. \label{90}
\end{equation}
However, Theorem 1 and Proposition 2 imply that $\sum_{i=1}^{\infty}\frac{\|\text{grad}_{\mathcal{M}}L(\mathbf{\Sigma}_{i})\|^{4}}{\|\mathbf{D}_{i}\|^{2}} < \infty$, which is a contradiction to \eqref{90}. 
Therefore, we otain the desired result $\liminf_{i\to\infty}\|\text{grad}_{\mathcal{M}}L(\mathbf{\Sigma}_{i})\| = 0$.

\section*{Appendix E\\Proof of Lemma 6}
In the first step of Algorithm 2, we compute $\textup{grad}_{\mathcal{M}}L(\mathbf{\Sigma})$, which is given by
\begin{equation}
\textup{grad}_{\mathcal{M}}L(\mathbf{\Sigma}_{i}) = ( \nabla_{\boldsymbol{\theta}}L(\boldsymbol{\theta}_{i}) - \Re \{\boldsymbol{\theta}_{i}^{*} \odot \nabla_{\boldsymbol{\theta}}L(\boldsymbol{\theta}_{i})\} \odot \boldsymbol{\theta}_{i} ) \oplus ( \nabla_{\mathbf{W}}L(\mathbf{W}_{i}) - \mathbf{W}_{i} \, \textup{ddiag}(\Re \{\mathbf{W}_{i}^{\text{H}} \nabla_{\mathbf{W}}L(\mathbf{W}_{i})\}) ). \label{91}
\end{equation}
Note that the numbers of flops required for computing the Euclidean gradients  $\nabla_{\boldsymbol{\theta}}L(\boldsymbol{\theta}_{i})$ and $\nabla_{\mathbf{W}}L(\mathbf{W}_{i})$ are $K^{2}N^{2}M + K^{2}N^{3}$ and $K^{2}N^{2}M + K^{2}M^{2}$, respectively.
Thus, the complexity $\mathcal{C}_{1}$ of computing $\textup{grad}_{\mathcal{M}}L(\mathbf{\Sigma}_{i})$ is 
\begin{equation}
\mathcal{C}_{1} = \mathcal{O}(K^{2}N^{2}M + K^{2}N^{3} + K^{2}M^{2}). \label{92}
\end{equation}
Using the Riemannian gradient $\textup{grad}_{\mathcal{M}}L(\mathbf{\Sigma}_{i})$, we then compute the RCG coefficient $\beta_{i} = \frac{\|\textup{grad}_{\mathcal{M}}L(\mathbf{\Sigma}_{i})\|^{2}}{\|\textup{grad}_{\mathcal{M}}L(\mathbf{\Sigma}_{i-1})\|^{2}}$.
The complexity $\mathcal{C}_{2}$ of computing $\beta_{i}$ is 
\begin{equation}
\mathcal{C}_{2} = \mathcal{O}(K^{2}N + K^{2}M). \label{93}
\end{equation}
Then the complexity $\mathcal{C}_{3}$ for updating the Riemannian conjugate direction $\mathbf{D}_{i}= -\textup{grad}_{\mathcal{M}}L(\mathbf{\Sigma}_{i}) + \beta_{i} \textbf{\textup{P}}_{\mathcal{T}_{\mathbf{\Sigma}_{i}}\mathcal{M}} (\mathbf{D}_{i-1})$ is 
\begin{equation}
\mathcal{C}_{3} = \mathcal{O}(K^{2}M + KN). \label{95}
\end{equation}

Next, we find out the step size $\alpha_{i}$ via the line search which consists of the computation of the following elements: 1) $\textbf{\textup{R}}_{\mathbf{\Sigma}_{i}}(\alpha_{i} \mathbf{D}_{i})$, 2) $\textup{grad}_{\mathcal{M}}L(\mathbf{\Sigma}_{i})$, 3) $L(\textbf{\textup{R}}_{\mathbf{\Sigma}_{i}}(\alpha_{i} \mathbf{D}_{i}))$, and 4) $L(\mathbf{\Sigma}_{i})$.
To compute the retraction $\textbf{\textup{R}}_{\mathbf{\Sigma}_{i}}(\alpha_{i} \mathbf{D}_{i})$, the required number of flops is $K^{2}M + N$.
To compute $\textup{grad}_{\mathcal{M}}L(\mathbf{\Sigma}_{i})$, the required number of flops is $K^{2}N^{2}M + K^{2}N^{3} +  K^{2}M^{2}$.
Also, the numbers of flops required for computing $L(\textbf{\textup{R}}_{\mathbf{\Sigma}_{i}}(\alpha_{i} \mathbf{D}_{i}))$ and $L(\mathbf{\Sigma}_{i})$ are both $K^{2}N^{2}M$.
Thus, the complexity $\mathcal{C}_{4}$ of the line search for updating $\alpha_{i}$ is 
\begin{equation}
\mathcal{C}_{4} = \mathcal{O}(K^{2}N^{2}M + K^{2}N^{3} +  K^{2}M^{2}). \label{97}
\end{equation}
Finally, the complexity $\mathcal{C}_{5}$ of updating $\mathbf{\Sigma}_{i}$ is 
\begin{equation}
\mathcal{C}_{5} = \mathcal{O}(K^{2}M + N). \label{98}
\end{equation}

In summary, the complexity $\mathcal{C}_{(\boldsymbol{\theta}, \mathbf{W})}$ of Algorithm 2 is 
\begin{equation}
\mathcal{C}_{(\boldsymbol{\theta}, \mathbf{W})} = \mathcal{C}_{1} +  \mathcal{C}_{2} +  \mathcal{C}_{3} +  \mathcal{C}_{4} +  \mathcal{C}_{5} = \mathcal{O}(K^{2}N^{2}M + K^{2}N^{3} +  K^{2}M^{2}). \label{99}
\end{equation}

\ifCLASSOPTIONcaptionsoff
  \newpage
\fi

\bibliographystyle{IEEEtran}
\bibliography{reference}

\end{document}